\documentclass[11.5pt,onecolumn,twoside]{article}
\usepackage{fullpage}
\usepackage{charter} 
\usepackage{titlesec}
\usepackage[colorlinks=true,linkcolor=black,anchorcolor=black,citecolor=black,filecolor=black,menucolor=black,runcolor=black,urlcolor=black]{hyperref}       
\usepackage{url}            
\usepackage{booktabs}       
\usepackage{amsfonts}       
\usepackage{nicefrac}       
\usepackage{microtype}      
\usepackage{xcolor}         
\usepackage{bm}         
\usepackage{graphicx}		
\usepackage{amsmath} 
\usepackage{cite} 
\usepackage{diagbox}
\usepackage{enumitem}
\usepackage{siunitx}
\usepackage{multirow}
\usepackage{tabularx}
\usepackage{hhline}
\usepackage{arydshln}		
\usepackage{xcolor}
\usepackage{mathtools} 
\usepackage{amsthm}			
\usepackage{amssymb}			

\usepackage{amssymb}        
\usepackage{algorithm}
\usepackage{algpseudocode}
\definecolor{lightgreen}{rgb}{.9,1,.9}

\newcolumntype{L}[1]{>{\raggedright\arraybackslash}p{#1}}
\newcolumntype{C}[1]{>{\centering\arraybackslash}p{#1}}
\newcolumntype{R}[1]{>{\raggedleft\arraybackslash}p{#1}}

\theoremstyle{plain} 

\newtheorem{theorem}{Theorem}
\newtheorem{lemma}{Lemma}

\newtheorem{assumption}{Assumption}

\def\defn{\,\coloneqq\,}
\def\Fix{{\mathsf{Fix}}}

\def\R{\mathbb{R}}
\def\E{\mathbb{E}}

\def\zerobm{{\bm{0}}}

\def\ebm{{\bm{e}}}

\def\xbm{{\bm{x}}}
\def\zbm{{\bm{z}}}
\def\ybm{{\bm{y}}}

\def\zbm{{\bm{z}}}
\def\sbm{{\bm{s}}}
\def\abm{{\bm{a}}}
\def\bbm{{\bm{b}}}

\def\xbfhat{{\widehat{\mathbf{x}}}}

\def\xbfhat{{\widehat{\bm{x}}}}

\def\Tsfhat{{\widehat{{\mathsf{T}}}}}

\def\Fsfhat{{\widehat{\mathsf{F}}}}
\def\ghat{{\widehat{{g}}}}

\def\Abm{{\bm{A}}}
\def\Dbm{{\bm{D}}}
\def\Pbm{{\bm{P}}}
\def\Sbm{{\bm{S}}}
\def\Fbm{{\bm{F}}}
\def\Ibm{{\bm{I}}}

\def\phibm{{\bm{\phi}}}

\def\thetabm{{\bm{\theta }}}
\def\ellhat{{\hat{\ell}}}

\def\Dsf{{\mathsf{D}}}
\def\Gsf{{\mathsf{G}}}

\def\Isf{{\mathsf{I}}}
\def\Rsf{{\mathsf{R}}}
\def\Tsf{{\mathsf{T}}}
\def\Tsf{{\mathsf{T}}}

\def\Dsf{{\mathsf{D}}}
\def\Fsf{{\mathsf{F}}}
\def\Hsf{{\mathsf{H}}}

\def\Gsf{{\mathsf{G}}}
\def\Isf{{\mathsf{I}}}

\def\Hsf{{\mathsf{H}}}

\def\bbmbar{{\overline{\bm{b}}}}
\def\xbmast{{\bm{x}^\ast}}
\def\xbmbar{{\overline{\bm{x}}}}

\def\xbmhat{{\widehat{\bm{x}}}}

\def\argmin{\mathop{\mathsf{arg\,min}}} 

\newcommand*\samethanks[1][\value{footnote}]{\footnotemark[#1]}
\title{Online Deep Equilibrium Learning for\\Regularization by Denoising}

\author{
Jiaming~Liu$^{\footnotesize 1,}$\thanks{These authors contributed equally.}{\;\,}, Xiaojian~Xu$^{\footnotesize 2,}$\samethanks{\;\,}, Weijie~Gan$^{\footnotesize 2}$, \\ 
Shirin~Shoushtari$^{\footnotesize 1}$,~and~Ulugbek~S.~Kamilov$^{\footnotesize 1, 2}$\\
\emph{\footnotesize $^{\footnotesize 1}$Department of Electrical and Systems Engineering,~Washington University in St.~Louis, MO 63130, USA}\\
\emph{\footnotesize $^{\footnotesize 2}$Department of Computer Science and Engineering,~Washington University in St.~Louis, MO 63130, USA}}
\begin{document}
\date{}
\maketitle

\begin{abstract}
Plug-and-Play Priors (PnP) and Regularization by Denoising (RED) are widely-used frameworks for solving imaging inverse problems by computing fixed-points of operators combining physical measurement models and learned image priors. While traditional PnP/RED formulations have focused on priors specified using image denoisers, there is a growing interest in learning PnP/RED priors that are end-to-end optimal. The recent Deep Equilibrium Models (DEQ) framework has enabled memory-efficient end-to-end learning of PnP/RED priors by implicitly differentiating through the fixed-point equations without storing intermediate activation values.  However, the dependence of the computational/memory complexity of the measurement models in PnP/RED on the total number of measurements leaves DEQ impractical for many imaging applications. We propose ODER as a new strategy for improving the efficiency of DEQ through stochastic approximations of the measurement models. We theoretically analyze ODER giving insights into its convergence and ability to approximate the traditional DEQ approach. Our numerical results suggest the potential improvements in training/testing complexity due to ODER on three distinct imaging applications.
\end{abstract}

\section{Introduction}

\vspace{-0.5em}
There has been considerable recent interest in using \emph{deep learning (DL)} in the context of imaging inverse problems (see recent reviews~\cite{McCann.etal2017, Lucas.etal2018, Ongie.etal2020}). Instead of explicitly defining a regularizer, the traditional DL approach is based on training a \emph{convolutional neural network (CNN)} architecture, such as U-Net~\cite{Ronneberger.etal2015}, to invert the measurement operator by exploiting the natural redundancies in the imaging data~\cite{DJin.etal2017, Kang.etal2017, Chen.etal2017, Sun.etal2018, han.etal2018}. \emph{Plug-and-Play Priors (PnP)}~\cite{Venkatakrishnan.etal2013} and \emph{Regularization by Denoising (RED)}~\cite{Romano.etal2017} are two well-known alternative approaches to the traditional DL that enable the integration of pre-trained CNN denoisers, such as DnCNN~\cite{Zhang.etal2017} or DRUNet~\cite{Zhang.etal2021b}, as image priors within iterative algorithms. When equipped with advanced CNN denoisers, PnP/RED provide excellent performance by exploiting both the implicit prior, characterized by a denoiser, and the measurement model~\cite{Chan.etal2016, Sreehari.etal2016, Kamilov.etal2017, Buzzard.etal2017, Reehorst.Schniter2019, Ryu.etal2019, Mataev.etal2019, Wu.etal2020, Liu.etal2020}. \emph{Deep Unfolding (DU)} is a related approach that interprets the iterations of an image recovery algorithm as layers of a neural network and trains it end-to-end in a supervised fashion. Unlike in PnP/RED, the CNN in DU is trained jointly with the measurement model, leading to an image prior optimized for a given inverse problem~\cite{zhang2018ista, Hauptmann.etal2018, Adler.etal2018, Aggarwal.etal2019, Hosseini.etal2019, Yaman.etal2020, Mukherjee.etal2021}. DU architectures, however, are usually limited to a small number of unfolded iterations due to the high computational and memory complexity of training.

Recent work on~\emph{Neural ODEs}~\cite{chen.etal2018, Dupont.etal2019, Kelly.etal2020} and~\emph{Deep Equilibrium Models (DEQ)}~\cite{Bai.etal2019, Winston.etal2019, Gilton.etal2021, Kawaguchi.etal2021, Fung.etal2021, Gurumurthy.etal2021} has shown the potential benefits of \emph{implicit neural networks} in a number of DL tasks. For example, DEQ was recently used to train CNN priors within PnP/RED iterations by differentiating through the fixed points of the corresponding iterations~\cite{Gilton.etal2021}. Training PnP/RED using DEQ is equivalent to training an infinite depth feedforward network integrating a physical measurement model and CNN prior. However, the training of such networks can still be a significant computational and memory challenge in applications that require processing of a large number of sensor measurements.  Specifically, the data-consistency layers in~\cite{Gilton.etal2021} are based on \emph{batch} processing, which means that the \emph{entire set of measurements} is processed at each layer. While this type of batch data processing is known to be suboptimal in traditional large-scale optimization~\cite{Bottou.Bousquet2007, Bertsekas2011, Kim.etal2013, Bottou.etal2018}, the issue has never been considered in the context of training of implicit networks such as those specified via PnP/RED iterations. 

This paper addresses this issue by proposing \emph{Online Deep Equilibrium RED (ODER)} as the first DEQ framework for inverse problems that adopts \emph{stochastic processing} of measurements within an implicit neural network. We argue that the proposed \emph{online} approach can improve training and testing efficiency compared to its \emph{batch} counterpart in a number of applications where the number of measurements is large. ODER can be implemented using the fixed-point iterations of RED by introducing stochastic approximations to the corresponding forward and backward DEQ passes. The CNN prior within ODER is trained end-to-end to remove artifacts due to the imaging system and stochastic processing. We present theoretical insights into the convergence of forward and backward passes of ODER, and show its ability to approximate learning using the traditional batch DEQ approach. We show the practical relevance of ODER by solving inverse problems in \emph{intensity diffraction tomography (IDT)}~\cite{Ling.etal18, Wu.etal2020}, sparse-view~\emph{computed tomography (CT)}~\cite{Kak.Slaney1988} and~accelerated parallel \emph{magnetic resonance imaging (MRI)}~\cite{Griswold2002, Uecker.etal2014}. Our numerical results show the ability of ODER to match the imaging quality of the batch DEQ learning at a fraction of complexity. Our work thus addresses an important gap in the current literature on PnP/RED, DU, and DEQ by providing an efficient framework applicable to a wide variety of imaging inverse problems.

All proofs and some technical details that have been omitted for space appear in the appendix, which also provides more background and simulations.
\vspace{-.5em}
\section{Background}

\textbf{Inverse problems.} Many imaging problems---such as IDT, CT, and MRI---can be formulated as an inverse problem involving the recovery of an image $\xbmast \in \R^n$ from noisy measurements $\ybm = \Abm\xbmast + \ebm$,
where $\Abm \in \R^{m \times n}$ is the measurement operator and $\ebm \in \R^m$ is the noise. A common approach to estimate $\xbmast$ is to solve an optimization problem
\begin{equation}
\label{Eq:OptimizationForInverseProblem}
\xbmhat = \argmin_{\xbm \in \R^n} \left\{g(\xbm) + h(\xbm)\right\},
\end{equation}
where $g$ is a data-fidelity term that quantifies consistency with the observed data $\ybm$ and $h$ is a regularizer that encodes prior knowledge on $\xbm$. A widely-used data-fidelity term and regularizer in inverse problems are $g(\xbm) = \frac{1}{2}\|\ybm-\Abm\xbm\|_2^2$ and the \emph{total variation (TV)} function $h(\xbm) = \tau \|\Dbm\xbm\|_1$, where $\Dbm$ is the gradient operator and $\tau > 0$ is the regularization parameter~\cite{Rudin.etal1992, Bioucas-Dias.Figueiredo2007, Beck.Teboulle2009a}. 

\vspace{+0.3em}
\noindent
\textbf{PnP, RED, and DU.}  PnP~\cite{Venkatakrishnan.etal2013, Sreehari.etal2016} and RED~\cite{Romano.etal2017} are two related classes of iterative algorithms that use \emph{additive white Gaussian noise (AWGN)} denoisers, such as BM3D~\cite{Dabov.etal2007} or DnCNN~\cite{Zhang.etal2017}, as priors for inverse problems (see the recent review~\cite{Kamilov.etal2022}). Since for general denoisers PnP/RED do not solve an optimization problem~\cite{Reehorst.Schniter2019}, it is common to interpret PnP/RED as fixed-point iterations of some high-dimensional operators. For example, given a denoiser $\Dsf_\thetabm: \R^n \rightarrow \R^n$ parameterized by a CNN with weights $\thetabm$, the \emph{steepest descent} variant of RED (SD-RED)~\cite{Romano.etal2017} can be written as
\begin{equation}
\label{Eq:REDItetation}
\xbm^k = \Tsf_\thetabm(\xbm^{k-1}) = \xbm^{k-1} - \gamma \Gsf_\thetabm(\xbm^{k-1}) \quad\text{with}\quad \Gsf_\thetabm(\xbm) \defn \nabla g(\xbm) + \tau(\xbm - \Dsf_\thetabm(\xbm))\;,
\end{equation}
where $g$ is the data-fidelity term, and  $\gamma, \tau > 0$ are the step size and the regularization parameters, respectively. SD-RED thus seeks to compute a fixed-point $\xbmbar \in \R^n$ of the operator $\Tsf$
\begin{equation}
\label{Eq:ZeroRED}
\xbmbar\in\Fix(\Tsf_\thetabm)\defn\{\xbm\in\R^n:\Tsf_\thetabm(\xbm)=\xbm\}\quad\Leftrightarrow\quad\Gsf_\thetabm(\xbmbar) = \nabla g(\xbmbar) + \tau(\xbmbar-\Dsf_\thetabm(\xbmbar))=\zerobm\;,
\end{equation}
The solutions of~\eqref{Eq:ZeroRED} balance the requirements to be both data-consistent (via $\nabla g$) and noise-free (via $(\Isf-\Dsf_\thetabm)$), which can be intuitively interpreted as finding an equilibrium between the physical measurement model and learned prior model. Remarkably, this heuristic of using denoisers not necessarily associated with any $h$ within an iterative algorithm exhibited great empirical success~\cite{Zhang.etal2017a, Metzler.etal2018, Dong.etal2019, Zhang.etal2019, Mataev.etal2019, Sun.etal2019b, Liu.etal2020, Ahmad.etal2020, Wei.etal2020, Xie.etal2021} and spurred a great deal of theoretical work on PnP/RED~\cite{Chan.etal2016, Meinhardt.etal2017, Buzzard.etal2017, Reehorst.Schniter2019, Ryu.etal2019, Sun.etal2018a, Tirer.Giryes2019, Teodoro.etal2019, Xu.etal2020, Sun.etal2021, Cohen.etal2020}. It is worth mentioning that there has been considerable effort in reducing the \emph{test-time} computational/memory complexity of PnP/RED by designing online and stochastic PnP/RED algorithms~\cite{Sun.etal2018a, Wu.etal2020, Tang.Davies2020, Sun.etal2021}. 

DU (also known as \emph{algorithm unrolling}) is a DL paradigm that has gained popularity due to its ability to systematically connect iterative algorithms and deep neural network architectures (see reviews in~\cite{Ongie.etal2020, Monga.etal2021}). Many PnP/RED algorithms have been turned into DU architectures by parameterizing the operator $\Dsf_{\thetabm}$ as a CNN with weights $\thetabm$, truncating the PnP/RED algorithm to a fixed number of iterations, and training the corresponding architecture end-to-end in a supervised fashion. Recent work has explored strategies for reducing the memory and computational complexity of training DU architectures~\cite{Kellman.etal2020, Liu.etal2021}, However, a key bottleneck in DU training is the necessity to store the intermediate activation values required for computing the backpropagation updates, which fundamentally limits the number of unfolding layers one can practically use in large-scale applications.

\vspace{+0.3em}
\noindent
\textbf{DEQ.}  DEQ~\cite{Bai.etal2019} is a recent method for training infinite-depth, weight-tied feedforward networks by analytically backpropagating through the fixed points using implicit differentiation. The DEQ output is specified implicitly as a fixed point of an operator $\Tsf_{\thetabm}$ parameterized by weights $\thetabm$ 
\begin{equation}
\label{Eq:DEQFixedPoint}
\xbmbar = \Tsf_{\thetabm}(\xbmbar)\; .
\end{equation}
The DEQ forward pass estimates $\xbmbar$ in~\eqref{Eq:DEQFixedPoint} by either running a fixed-point iteration or using an optimization algorithm. The DEQ backward pass produces gradients with respect to $\thetabm$ by implicitly differentiating through the fixed points without the knowledge of how they are estimated
\begin{align}
\label{Eq:DEQ1}
\ell(\thetabm)= \frac{1}{2}\|\xbmbar(\thetabm)-\xbmast\|_2^2\quad\Rightarrow\quad \nabla\ell(\thetabm) = (\nabla_\thetabm \Tsf_\thetabm(\xbmbar))^\Tsf \left(\Isf - \nabla_\xbm \Tsf(\xbmbar)\right)^{-\Tsf}(\xbmbar-\xbmast),
\end{align}
where $\ell$ is the loss function, $\xbmast$ is the training label, and $\Isf$ is the identity mapping.  The vector product with the inverse-Jacobian in~\eqref{Eq:DEQ1} can be approximated by solving the following fixed-point equation
\begin{align}
\label{Eq:DEQ2}
\bbmbar\defn\left(\Isf - \nabla_\xbm \Tsf_\thetabm(\xbmbar)\right)^{-\Tsf}(\xbmbar-\xbmast)\quad\Rightarrow\quad\bbmbar= \left(\nabla_\xbm \Tsf_\thetabm(\xbmbar)\right)^{\Tsf}\bbmbar + (\xbmbar-\xbmast)\,.
\end{align}
Recent work has also explored Jacobian-free DEQ by replacing the inverse-Jacobian with an identity mapping $\Isf$, leading to a faster training~\cite{Fung.etal2021}. 

The comparison of equations~\eqref{Eq:REDItetation},~\eqref{Eq:ZeroRED}, and~\eqref{Eq:DEQFixedPoint} highlights an elegant connection between PnP/RED and DEQ. This connection was explored in the recent work~\cite{Gilton.etal2021} by using DEQ for learning the weights of the CNN prior $\Dsf_\thetabm$ end-to-end within PnP/RED iterations. Within the framework of~\cite{Gilton.etal2021}, PnP/RED is used for the forward pass and a backward pass is obtained by using~\eqref{Eq:DEQ2} on the PnP/RED operators. Specifically, the CNN prior in SD-RED can be trained by running the backward pass using $\Tsf_\thetabm$ in~\eqref{Eq:REDItetation} 
\begin{align}
\label{Eq:BPUpdate}
\bbm^{k} = \Fsf(\bbm^{k-1}) = \left(\nabla_\xbm \Tsf_\thetabm(\xbmbar)\right)^\Tsf\bbm^{k-1} + (\xbmbar-\xbmast).
\end{align}

This work makes several new contributions to the existing literature on PnP/RED, DU, and DEQ. The focus is to explore the impact of approximating $\Tsf_\thetabm$ in~\eqref{Eq:DEQFixedPoint} with a ``simpler'' operator $\widehat{\Tsf}_\thetabm$. Following~\cite{Gilton.etal2021}, we focus on inverse problems by using PnP/RED operators of form~\eqref{Eq:REDItetation} that integrate the physical measurement models and learned CNN priors. We give algorithmic, theoretical, and numerical evidence that one can achieve significant memory/computational gains, while preserving the performance of the original DEQ approach in~\cite{Gilton.etal2021}. It is worth noting that the results here have the potential to generalize to many other implicit neural networks beyond those specified via PnP/RED.

\section{Online Deep Equilibrium Method}
\label{Sec:OnlineDRED}
\vspace{-0.5em}
We consider inverse problems where the data-fidelity term $g$ can be expressed as
\begin{equation}
\label{Eq:StochData}
g(\xbm) = \frac{1}{b}\sum_{i=1}^{b}g_i(\xbm),
\end{equation}
where each $g_i$ depends only on the subset $\ybm_i \in \R^{m_i}$ of the full measurements $\ybm \in \R^m$ as
\begin{equation*}
\R^m = \R^{m_1} \times \R^{m_2} \times \cdots \times \R^{m_b} \quad\text{with}\quad m = m_1 + m_2 + \cdots + m_b\;.
\end{equation*}
We are primarily interested in scenarios where the memory/computational complexity of the gradient $\nabla g$ is proportional to $b$. Thus, when $b\rightarrow \infty$, the memory and computational complexity of traditional DEQ to train the CNN prior within the batch PnP/RED algorithms becomes impractical. 

To decouple the computational/memory complexity of DEQ from $b$, we adopt \emph{online} processing of measurements, where $g$ is approximated using a minibatch of $w \ll b$ measurements
\begin{equation}
\label{Eq:SGDPhysParam}
\ghat(\xbm) = \frac{1}{w} \sum_{s = 1}^w g_{i_s}(\xbm)
\quad\Rightarrow\quad 
\nabla\ghat(\xbm) = \frac{1}{w} \sum_{s = 1}^w \nabla g_{i_s}(\xbm)
\quad\Rightarrow\quad 
\Hsf\ghat(\xbm) = \frac{1}{w} \sum_{s = 1}^w \Hsf g_{i_s}(\xbm)\;,
\end{equation}
where $\{i_1, \dots, i_w\}$ are i.i.d random variables selected uniformly from the set $\{1,\dots,b\}$. Note that~\eqref{Eq:SGDPhysParam} directly implies the \emph{unbiasedness} of the online gradient $\E\left[\nabla \ghat(\xbm)\right] = \nabla g(\xbm)$ and Hessian $\E\left[\Hsf\ghat(\xbm)\right] = \Hsf g(\xbm)$ with the expectations taken over the random indices $\{i_1, \dots, i_w\}$. 

\subsection{Forward Pass}

The forward-pass of ODER is performed as follows
\begin{equation}
\label{Eq:SGDFWUpdate}
\xbm^k = \Tsfhat_\thetabm(\xbm^{k-1}) = \xbm^{k-1} - \gamma (\nabla \ghat(\xbm^{k-1}) + \tau \Rsf_\thetabm(\xbm^{k-1})), \quad k =1, 2, \dots, K,
\end{equation}
where $\Rsf_{\thetabm}=\Isf-\Dsf_{\thetabm}$ is the residual of the CNN prior $\Dsf_{\thetabm}$. The residual $\Rsf_{\thetabm}$ takes artifact-corrupted images at the input and produces the corresponding artifacts at the output. Note how the ODER forward-pass is independent of $b$ since it uses a minibatch approximation $\nabla\ghat$ in~\eqref{Eq:SGDPhysParam}.

 It is worth mentioning that when considered separately from DEQ, the forward pass of ODER corresponds to the existing online RED algorithm~\cite{Wu.etal2019, Wu.etal2020}. The contribution of this work is thus not the forward pass, but the integration of the online forward pass into the DEQ framework via the online backward pass, resulting in a more scalable and flexible DEQ framework for inverse problems.

\subsection{Backward Pass}

The backward pass of ODER for the MSE loss in~\eqref{Eq:DEQ1} is performed as follows
\vspace{-.5em}
\begin{align}
\label{Eq:SGDBPUpdate}
\bbm^k = \Fsfhat(\bbm^{k-1}) = \left[\nabla_\xbm \Tsfhat_\thetabm\left(\xbm^K\right)\right]^\Tsf\bbm^{k-1} + (\xbm^K-\xbmast), \quad k = 1, 2, \dots, K,
\end{align}
starting from $\bbm^0 = \zerobm$, where $\xbm^K$ is the final iterate of the forward pass~\eqref{Eq:SGDFWUpdate} at iteration $K \geq 1$. In both traditional and online backward passes, conventional auto-differentiation tools enable the computation of the Jacobian-vector products in~\eqref{Eq:BPUpdate} and~\eqref{Eq:SGDBPUpdate}. However, the key difference is that the computational complexity of ODER does not depend on the total number of measurements $b$.

\subsection{ODER Learning}

ODER seeks to minimize the MSE loss over $p \geq 1$ training samples
\begin{align}
\label{Eq:Loss}
\ell(\thetabm) = \frac{1}{p}\sum_{j = 1}^p\ell_j(\thetabm)\quad\text{with}\quad\ell_j(\thetabm)=\frac{1}{2}\|\xbmbar_j(\thetabm)-\xbm_j^{\ast}\|_2^2 \, ,
\end{align}
using approximate gradients $\nabla \hat{\ell}_j(\thetabm)$ computed via the online forward and backward passes that are independent of $b$. Here, $\thetabm$ denotes the weights of the CNN prior, $\xbm_j^{\ast}$ is the $j$th training label, and $\xbmbar_j(\thetabm)$ is the fixed-point of the full-batch SD-RED algorithm~\eqref{Eq:ZeroRED}. The backward pass of ODER can be integrated within any gradient-based optimizer, such as the~\emph{stochastic gradient descent (SGD)}. At training iteration $t \geq 1$ we generate two sets of independent random variables. First, the index $j_t$ is selected uniformly at random from $\{1,\dots, p\}$, then online forward and backward passes are computed using the measurement models in~\eqref{Eq:SGDPhysParam}. We can thus express the SGD update rule as follows
\begin{align}
\label{Eq:TrainUpdate}
\thetabm^{t+1} = \thetabm^t - \beta\nabla\hat{\ell}_{j_t}(\thetabm^t)\quad\text{with}\quad\nabla\hat{\ell}_{j_t}(\thetabm^t)= \left[\nabla_\thetabm \Tsfhat_{\thetabm^t}({\xbm_{j_t}^K})\right]^{\Tsf}\bbm_{j_t}^K, 
\end{align}
where $\beta>0$ is the SGD learning rate, $\xbm_{j_t}^K$ and $\bbm_{j_t}^K$ are the final iterates of the ODER forward and backward passes at the training index $j_t$ after $K \geq 1$ iterations.

\section{Theoretical Analys}
\label{Sec:TheoreticalAnalys}
Our theoretical analysis relies on a set of explicit assumptions serving as sufficient conditions. The proofs of all the theorems will be provided in the supplement. 
\begin{assumption}
\label{As:LipConObj}
Each $g_i$ is twice continuously differentiable and convex. There exists $\lambda > 0$ such that each gradient $\nabla g_i$ and Hessian $\Hsf g_i$ are $\lambda$-Lipschitz continuous.
\end{assumption}
The fact that $g$ is twice continuously differentiable is needed for the backward pass. The assumption that all the Lipschitz constants are the same is only needed to streamline mathematical exposition.

\begin{assumption}
\label{As:LipConDen}
$\Dsf_\thetabm(\xbm)$ is continuously differentiable with respect to $\thetabm$ and $\xbm$. There exists $\alpha > 0$ such that $\Dsf_\thetabm(\xbm)$, $\nabla_\xbm \Dsf_\thetabm(\xbm)$, and $\nabla_\thetabm \Dsf_\thetabm(\xbm)$ are $\alpha$-Lipschitz continuous with respect to $\thetabm$ and $\xbm$. Finally, we also assume that $\Dsf_\thetabm$ is a contraction, which means that there exists $\kappa < 1$ such that
\begin{equation*}
\|\Dsf_\thetabm(\zbm)-\Dsf_\thetabm(\ybm)\|_2  \leq \kappa \|\zbm-\ybm\|_2, \quad\forall \zbm, \ybm \in \R^n.
\end{equation*}
\end{assumption}
Since $\Dsf_\thetabm$ is a CNN, its differentiability is a standard assumption. The contractive $\Dsf_\thetabm$ and convex $g$, ensure that $\Tsf_\thetabm$ is a contraction, enabling provable convergence of the forward and backward passes. The design of contractive $\Tsf_\thetabm$ is a common PnP/RED strategy to ensure convergence~\cite{Ryu.etal2019, Gilton.etal2021, Sun.etal2021}.
\begin{assumption}
\label{As:BoundedIterates}
There exists $R > 0$ such that for all $\xbmbar \in \Fix(\Tsf)$ and $\bbmbar \in \Fix(\Fsf)$, we have $\|\xbm^k - \xbmbar\|_2 \leq R$ and $\|\bbm^k - \bbmbar\|_2 \leq R$ for all $k \in \{1, \dots, K\}$.
\end{assumption}
The existence of the bound $R$ is reasonable, as many images have bounded pixel values. Similarly, the bound on $\bbm^k$ is also reasonable for ensuring bounded DEQ gradients.
\begin{assumption}
\label{As:BoundedVariance}
There exists $\nu > 0$ such that for all $\xbm \in \R^n$, we have
\begin{equation*}
\E\left[\|\nabla g(\xbm) - \nabla \ghat(\xbm)\|_2^2\right] \leq \frac{\nu^2}{w}\quad\text{and}\quad \E\left[\|\Hsf g(\xbm) - \Hsf \ghat(\xbm)\|_2^2\right] \leq \frac{\nu^2}{w},
\end{equation*}
where the expectations are taken over $\{i_1, \dots, i_w\}$.
\end{assumption}
The variance bounds are standard in stochastic algorithms. The variance bounds on the gradient and Hessian approximations are thus reasonable in this context. The decrease of the bounds for higher values of $w$ is natural since $\ghat$ is an unbiased estimator of $g$ obtained averaging $w$ independent terms. 

\begin{theorem}
\label{Thm:Thm1}
Run the forward pass of ODER for $k \geq 1$ iterations under Assumptions~\ref{As:LipConObj}-\ref{As:BoundedVariance} using the step size $0 < \gamma < 1/(\lambda+\tau)$. Then, the sequence of forward pass iterates satisfies
\begin{equation}
\E\left[\|\xbm^k - \xbmbar\|_2\right] \leq  \eta^k R + \frac{\gamma \nu}{(1-\eta)\sqrt{w}},
\end{equation}
for some constant $0 < \eta < 1$ where $\xbmbar \in \Fix(\Tsf)$.
\end{theorem}
Theorem~\ref{Thm:Thm1} is a variation on the convergence results for online RED/PnP~\cite{Sun.etal2018a, Wu.etal2020, Tang.Davies2020, Sun.etal2021}, showing that the forward pass converges to $\xbmbar \in \Fix(\Tsf)$ up to an error term that can be controlled via $\gamma$ and $w$.

\begin{theorem}
\label{Thm:NeumannOnRED}
Run the backward pass of ODER for $k \geq 1$ iterations under Assumptions~\ref{As:LipConObj}-\ref{As:BoundedVariance} from $\bbm^0 = \zerobm$ using the step-size $0 < \gamma < 1/(\lambda+\tau)$. Then, the sequence of backward pass iterates satisfies
\begin{equation}
\E\left[\|\bbm^k-\bbmbar\|_2\right] \leq B_1 \eta^k + \frac{B_2}{\sqrt{w}},
\end{equation}
where $0 < \eta < 1$, $B_1 > 0$ and $B_2 > 0$ are constants independent of $k$ and $w$, and $\bbmbar \in \Fix(\Fsf)$.
\end{theorem}
Theorem~\ref{Thm:NeumannOnRED} shows that the online backward pass in expectation converges to $\bbmbar$ up to an error term that can be controlled via $w$. The complete expressions for constants $B_1$ and $B_2$ are in the proof. 

\begin{assumption}
\label{As:LipschitzCon}
Function $\,\ell\,$ has a global minimizer $\;\thetabm^{\ast}$ and has a $L$-Lipschitz continuous gradient, which means that for all $\thetabm, \phibm$, we have  $\|\nabla \ell(\thetabm) - \nabla \ell(\phibm)\|_2 \leq L\|\thetabm-\phibm\|_2$.
\end{assumption}
\begin{assumption}
\label{As:UnbiasedAssumptionF}
The loss function in~\eqref{Eq:Loss} and indices $\{j_t\}$ in~\eqref{Eq:TrainUpdate} are such that
\begin{equation*}
\E\left[\nabla\ell_{j_t}(\thetabm)\right]\,=\,\nabla \ell(\thetabm)\quad \text{and} \quad \E\left[\|\nabla\ell_{j_t}(\thetabm) - \nabla\ell(\thetabm)\|_2^2\right]\,\leq\, \sigma^2,
\end{equation*}
where the expectations are taken with respect to the random index uniformly as $j_t \in \{1,\dots, p\}$.
\end{assumption}
The existence of a minimizer and the Lipschitz continuity of the loss gradient are standard assumptions in the literature~\cite{Nesterov2004, Jain.Kar2017, Bottou.etal2018}.  Note that we do \emph{not} assume that the training loss $\ell$ is convex.
 Assumption~\ref{As:UnbiasedAssumptionF} is the standard assumption used in the analysis of SGD.
\begin{theorem}
\label{Thm:MainConvRes}
Train ODER using SGD for $T \geq 1$ iterations under Assumptions~1-6 using the step-size parameters $0 < \beta \leq 1/L$ and the minibatch size $w \geq 1$. Select a large enough number of forward and backward pass iterations $K \geq 1$ to satisfy $0<\eta^K\leq1/{\sqrt{w}}$. Then, we have that
\begin{equation*}
\label{EqTheorem1}
\frac{1}{T}\sum_{t=0}^{T-1}\E\left[\|\nabla\ell(\thetabm^t)\|_2^2\right]\leq\frac{2(\ell(\thetabm^0)-\ell(\thetabm^{\ast}))}{\beta T} + \frac{C_1}{\sqrt{w}} + \beta C_2\,.
\end{equation*}
where $C_1>0$ and $C_2>0$ are constants independent of $T$ and $w$.
\end{theorem}
The complete expressions for constants $C_1$ and $C_2$ are in the proof. Theorem~\ref{Thm:MainConvRes} allows us to understand the ability of the iterates generated using~\eqref{Eq:TrainUpdate} to approximate the stationary points of the desired loss~\eqref{Eq:Loss}. The error terms in the bound depend on the training step-size $\beta$ and the minibatch size $w$, both of which can be controlled during training. In summary, we presented several theoretical results that give insights into ODER by stating explicit sufficient conditions and approximation bounds.

\vspace{-0.3em}
\section{Numerical Evaluation}
\label{Sec:NumericalValidation}
We numerically validate ODER in the context of three computational imaging modalities: IDT, sparse-view CT, and parallel MRI. Our goal is to both (a) empirically evaluate the performance of ODER and (b) highlight its effectiveness for processing a large number of measurements. We adopt $\ell_2$-norm loss $g(\xbm)=\frac{1}{2}\|\ybm-\Abm\xbm\|_2^2$ as the data-fidelity term for all three imaging modalities. 

ODER is compatible with any CNN architecture used to implement $\Dsf_\thetabm$. We use a~\emph{tiny} U-Net architecture~\cite{Liu.etal2021} for ODER and the traditional RED (DEQ)~\cite{Gilton.etal2021}. We have added spectral normalization~\cite{Miyato.etal2018} to all the layers of CNN for stability (see the supplement for the numerical evaluation of the contractiveness of $\Tsf_\thetabm$ on all three modalities). Similar to~\cite{Gilton.etal2021}, the CNN prior of ODER and RED (DEQ) are initialized using pre-trained denoisers. During the training of both ODER and RED (DEQ), we use the \emph{Nesterov acceleration}~\cite{Nesterov2004} for the forward pass and \emph{Anderson acceleration}~\cite{anderson1965} for the backward pass. We also adopt the stopping criterion from~\cite{Gilton.etal2021, bai.etal2022} by setting residual tolerance to $10^{-3}$ for both forward and backward iterations (see supplement for additional details).
\begin{figure}[t]
\centering\includegraphics[width=\textwidth]{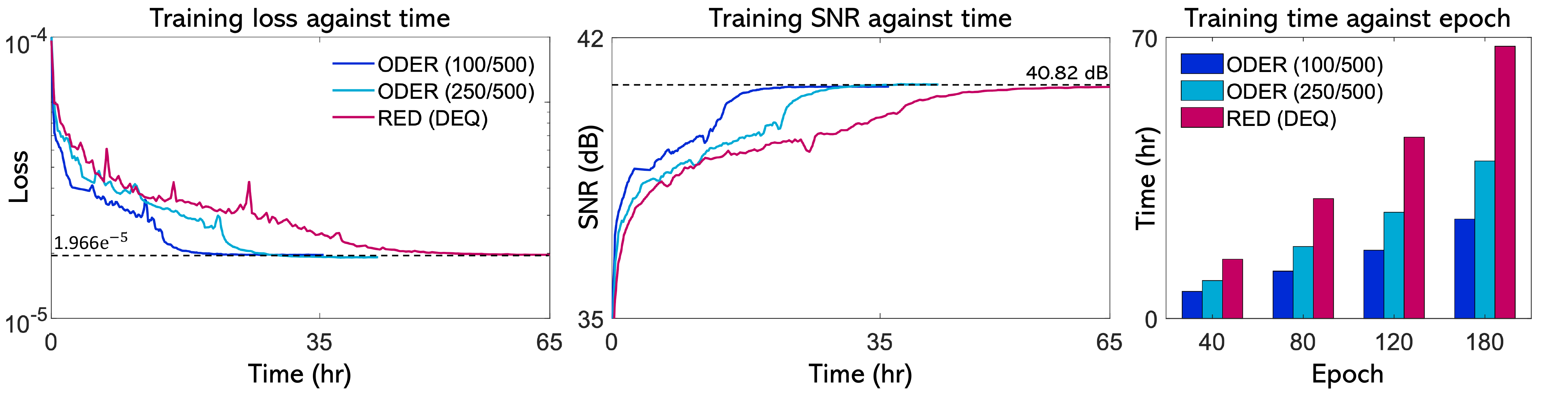}
\caption{~\emph{Quantitative evaluation of ODER on IDT for two minibatch sizes $w \in \{100, 250\}$ used at each step of the network against RED (DEQ) using the full batch of $b = 500$ measurements. The left figure plots the loss against time in hours for different values of $w$ evaluated on the training set. The middle and right figures plot the corresponding SNR against time and the amount of time required to reach a certain epoch for different values of $w$. By using minibatches $ 1\leq w \leq b$, ODER can achieve nearly $2.5 \times$ improvement in training time over RED (DEQ) for a similar final imaging quality.}}
\label{Fig:IDTtraining}
\vspace{-.5em}
\end{figure}
\begin{table}[t]
  \caption{IDT image recovery for different input SNR (dB) values on images from~\cite{aksac2019brecahad}.  We also present model size and per-iteration memory usage for the measurements, and average test-times.}
      \centering
      \renewcommand\arraystretch{1.2}
      {\footnotesize
      \scalebox{0.96}{
      \begin{tabular*}{12.78cm}{L{80pt}||C{30pt}lC{80pt}lC{80pt}lC{80pt}lC{80pt}lC{80pt}lC{80pt}l}
          \hline
          \multirow{2}{4em}{\textbf{Method}}& \multicolumn{3}{c}{\textbf{Input SNR (dB)}} & \multicolumn{2}{c}{\textbf{Size} } &\multicolumn{2}{c}{\textbf{Time}} \\
          \cline{2-9}
          &\multicolumn{1}{c}{15}  & \multicolumn{1}{c}{20} & \multicolumn{1}{c}{25} &  \multicolumn{1}{c}{Model} &  \multicolumn{1}{c}{Meas.} & \multicolumn{1}{c}{CPU} & \multicolumn{1}{c}{GPU}\\\hline\hline
          \textbf{TV}       & {38.34}    & {38.77}   & \multicolumn{1}{c:}{38.85} & \multicolumn{1}{c}{-----} &  \multicolumn{1}{c:}{3.56 GB} & \multicolumn{1}{c}{215.3s}   & \multicolumn{1}{c}{32.24s} \\
          \textbf{U-Net}       & {38.35}      & {38.89}       & \multicolumn{1}{c:}{39.02}     & \multicolumn{1}{c}{118.2 MB}     & \multicolumn{1}{c:}{-----} &  \multicolumn{1}{c}{2.811s}   & \multicolumn{1}{c}{0.089s} \\        
           \textbf{ISTA-Net+}       & 38.37   & 38.94    & \multicolumn{1}{c:}{39.27}    & \multicolumn{1}{c}{1.21 MB}     &  \multicolumn{1}{c:}{3.56 GB} &  \multicolumn{1}{c}{7.081s}   & \multicolumn{1}{c}{0.216s} \\
            \textbf{SGD-Net (100)}    & \textcolor{black}{39.62}     & {40.26}     & \multicolumn{1}{c:}{\textcolor{black}{40.47}}      & \multicolumn{1}{c}{29.7 MB}   &  \multicolumn{1}{c:}{0.71GB}  &  \multicolumn{1}{c}{6.697s}   & \multicolumn{1}{c}{0.207s}  \\
           \textbf{RED (Denoising)}      & {39.52}   & {40.04}     & \multicolumn{1}{c:}{40.41}    & \multicolumn{1}{c}{118.2 MB}   &  \multicolumn{1}{c:}{3.56 GB} &  \multicolumn{1}{c}{285.5s}   & \multicolumn{1}{c}{7.528s} \\
           \cdashline{1-9}
           \textbf{ODER (100)}      & 40.28    & 41.42     &  \multicolumn{1}{c:}{41.94}     &  \multicolumn{1}{c}{29.7 MB}  &  \multicolumn{1}{c:}{0.71 GB} & \multicolumn{1}{c}{63.31s}   & \multicolumn{1}{c}{2.051s} \\
           \textbf{ODER (250)}      &  \textcolor{black}{\textbf{40.57}}    & 41.50     &  \multicolumn{1}{c:}{\textcolor{black}{\textbf{41.96}}}     &  \multicolumn{1}{c}{29.7 MB}  &  \multicolumn{1}{c:}{1.76 GB} &  \multicolumn{1}{c}{118.7s}   & \multicolumn{1}{c}{3.628s}   \\
           \textbf{RED (DEQ)}       & 40.54  & \textcolor{black}{\textbf{41.51} }    &  \multicolumn{1}{c:}{41.95}     & \multicolumn{1}{c}{29.7 MB}   &  \multicolumn{1}{c:}{3.56 GB} & \multicolumn{1}{c}{202.3s} &   \multicolumn{1}{c}{6.362s} \\\hline
      \end{tabular*}}
      }
  \label{Tab:table1}
  \end{table}

For reference we include several other well-known baseline methods, including TV~\cite{Rudin.etal1992}, U-Net~\cite{Ronneberger.etal2015} and ISTA-Net$^{+}$~\cite{zhang2018ista}. We also include the unfolded \emph{RED (Unfold)}~\cite{Liu.etal2021} and the traditional \emph{RED (Denoising)}~\cite{Romano.etal2017} to illustrate the improvements due to DEQ.   TV is an iterative method that does not require training, while other methods are all DL-based with publicly available implementations. We use the U-Net architecture in~\cite{Ronneberger.etal2015} as the AWGN denoiser for RED, while we use the same tiny U-Net for RED (Unfold) as in RED (DEQ). For each imaging modality, we trained the denoiser in RED (Denoising) for AWGN removal at five noise levels corresponding to $\sigma\in$\{2, 7, 5, 10, 15\}. For each experiment, we select the denoiser achieving the highest SNR. In all the experiments, we train ODER and RED (DEQ) using the same training strategy and parameter initialization settings. We use \texttt{fminbound} in the \texttt{scipy.optimize} toolbox to  identify the optimal regularization parameters for TV, RED (Denoising), ODER and RED (DEQ) at the inference time.
\begin{figure}[t]
\centering\includegraphics[width=1.01\textwidth]{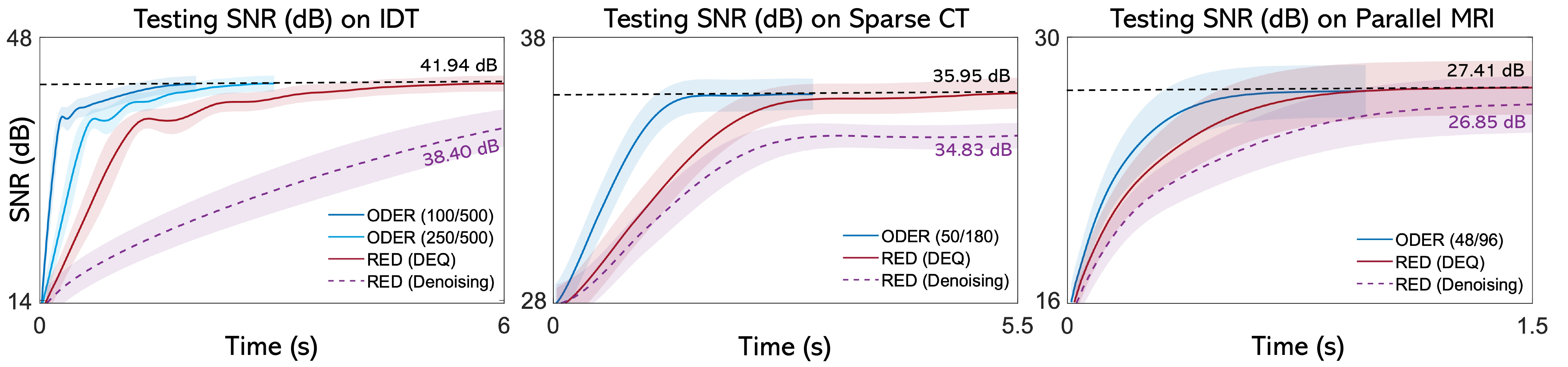}
\caption{~\emph{Illustration of the convergence speed of ODER, RED (DEQ) and RED (Denoising) for three imaging applications. \textbf{Left:} IDT with the full batch of $b=500$ measurements under 25 dB input SNR. \textbf{Middle:} Sparse-view CT with $b=180$ projection views. \textbf{Right:} Parallel MRI at  $20\%$ sampling with $b=96$ simulated coil sensitivity maps. ODER achieves $1.4\times\sim3\times$ speedup over RED (DEQ) at inference time without significant degradation in accuracy across three problems.}}
\label{Fig:testsnr}
\vspace{-.5em}   
\end{figure}
 \begin{table}[t]
\caption{Sparse-view CT image recovery in terms of SNR (dB) and SSIM on test images from~\cite{mccollough2016tu}. The last two columns provide the average test-times for a $512\times512$ image using 180 views.}
    \centering
    \renewcommand\arraystretch{1.2}
    {\footnotesize
    \scalebox{0.96}{
    \begin{tabular*}{14.58cm}{L{74pt}||C{32pt}lC{32pt}lC{32pt}lC{32pt}lC{32pt}lC{32pt}lC{32pt}l}
        \hline
        \multirow{2}{4em}{\textbf{Method}}& \multicolumn{6}{c}{\textbf{Projection Views} } & \multicolumn{2}{c}{\textbf{Time}} \\
        \cline{2-9}
        &\multicolumn{2}{c}{90}  & \multicolumn{2}{c}{120} & \multicolumn{2}{c}{180} & CPU & GPU\\\hline\hline
        \textbf{TV}       & {29.44}            & \multicolumn{1}{c:}{0.9688}            & {30.27}            & \multicolumn{1}{c:}{0.9731}     & 31.33 & \multicolumn{1}{c:}{0.9771}            & \multicolumn{1}{c}{ 768.1s}   & \multicolumn{1}{c}{ 15.61s} \\
        \textbf{U-Net}       & {33.05}            & \multicolumn{1}{c:}{0.9741}            & {34.02}            & \multicolumn{1}{c:}{0.9790}     & 35.11 & \multicolumn{1}{c:}{0.9815}            & \multicolumn{1}{c}{4.014s}  & \multicolumn{1}{c}{ 0.056s} \\        
         \textbf{ISTA-Net+}       & 32.15            & \multicolumn{1}{c:}{0.9706}           & 33.38           & \multicolumn{1}{c:}{0.9755}     & 34.83 & \multicolumn{1}{c:}{0.9812}           & \multicolumn{1}{c}{ 37.38s}  & \multicolumn{1}{c}{0.344s}  \\
          \textbf{RED (Unfold)}    & \textcolor{black}{33.97}            & \multicolumn{1}{c:}{\textcolor{black}{0.9753}}            & {35.01}            & \multicolumn{1}{c:}{0.9824}   & 35.78  & \multicolumn{1}{c:}{0.9835}           &\multicolumn{1}{c}{29.93s}   & \multicolumn{1}{c}{0.256s} \\
         \textbf{RED (Denoising)}      & {32.64}            & \multicolumn{1}{c:}{0.9708}            & {33.60}            & \multicolumn{1}{c:}{0.9789}   & 34.83 & \multicolumn{1}{c:}{0.9807}            & \multicolumn{1}{c}{498.5s}   & \multicolumn{1}{c}{5.549s} \\
         \cdashline{1-9}
         \textbf{ODER}      & 34.40           & \multicolumn{1}{c:}{\textcolor{black}{{0.9824}}}          & 35.12           &  \multicolumn{1}{c:}{0.9841}  & 35.91 & \multicolumn{1}{c:}{0.9859}            &\multicolumn{1}{c}{334.1s}   &\multicolumn{1}{c}{3.113s} \\
         \textbf{RED (DEQ)}       & \textcolor{black}{\textbf{34.61}}           &  \multicolumn{1}{c:}{\textbf{0.9826}}           & \textcolor{black}{\textbf{35.26}}            & \multicolumn{1}{c:}{\textcolor{black}{\textbf{0.9845}}}   & \textcolor{black}{\textbf{\;35.95} } &\multicolumn{1}{c:}{\textcolor{black}{\textbf{0.9861}}}           & \multicolumn{1}{c}{616.1s}             &  \multicolumn{1}{c}{5.466s} \\\hline
    \end{tabular*}}
    }
\label{Tab:table2}
\end{table}

\begin{figure}[t]
\centering\includegraphics[width=0.95\textwidth]{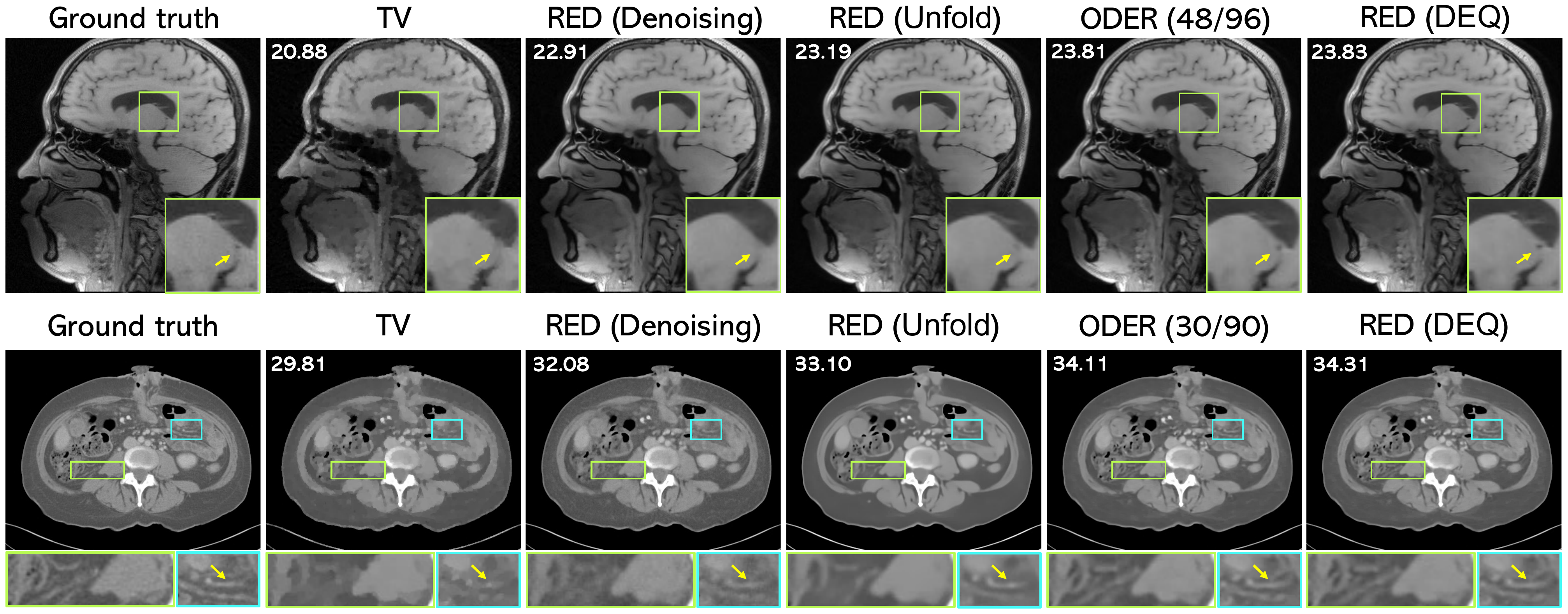}
\caption{~\emph{Visual evaluation of several well-known methods on two imaging problems: (top) Reconstruction of a brain image from its radial Fourier measurements at $10\%$ sampling with $b=96$ simulated coil sensitivity maps; (bottom) Reconstruction of a body CT image from $b=90$ projection views. Note the similar performance of ODER and RED (DEQ), and the improvement over RED (Denoising) /RED (Unfold) due to the usage of DEQ learning. Best viewed by zooming in the display.}}
\label{Fig:vsMRICT}
\end{figure}
 \begin{table}[t]
\caption{Average SNR (dB), SSIM, and running times for several methods on MRI images. The last two columns provide the average test-times for a $320\times320$ image using 96 simulated coils.}
    \centering
    \renewcommand\arraystretch{1.2}
    {\footnotesize
    \scalebox{1}{
    \begin{tabular*}{14.62cm}{L{75pt}||C{32pt}lC{32pt}lC{32pt}lC{32pt}lC{32pt}lC{32pt}lC{32pt}l}
        \hline
        \multirow{2}{4em}{\textbf{Method}}& \multicolumn{4}{c}{\textbf{MRI Set1~\cite{zhang2018ista}} } & \multicolumn{2}{c}{\textbf{MRI Set2~\cite{knoll2020fastmri}}} & \multicolumn{2}{c}{\textbf{Time}} \\
        \cline{2-9}
        &\multicolumn{2}{c}{10\%}  & \multicolumn{2}{c}{20\%} & \multicolumn{2}{c}{10\%} & CPU & GPU\\\hline\hline
        \textbf{TV}       & {20.88}            & \multicolumn{1}{c:}{0.9059}            & {24.87}            & \multicolumn{1}{c:}{0.9445}     & 24.84 & \multicolumn{1}{c:}{0.9674}            & 122.2s   & 7.591s \\
         \textbf{U-Net}       & 23.07            & \multicolumn{1}{c:}{0.9329}           & 26.42           & \multicolumn{1}{c:}{0.9562}     & 26.04 & \multicolumn{1}{c:}{0.9712}           & 0.683s   & 0.011s  \\        
         \textbf{ISTA-Net+}       & 22.95            & \multicolumn{1}{c:}{0.9298}           & 26.31           & \multicolumn{1}{c:}{0.9546}     & 25.82 & \multicolumn{1}{c:}{0.9693}           & 8.993s   & 0.264s  \\
          \textbf{RED (Unfold)}    & \textcolor{black}{23.37}            & \multicolumn{1}{c:}{\textcolor{black}{0.9363}}            & {26.81}            & \multicolumn{1}{c:}{0.9591}   & 26.37  & \multicolumn{1}{c:}{0.9744}           &8.744s   & 0.231s \\
         \textbf{RED (Denoising)}      & {23.29}            & \multicolumn{1}{c:}{0.9352}            & {26.85}            & \multicolumn{1}{c:}{0.9598}   & 26.42 & \multicolumn{1}{c:}{0.9748}            & 272.4s   & 7.511s \\
         \cdashline{1-9}
         \textbf{ODER}      & 24.08           &  \multicolumn{1}{c:}{0.9442}          & 27.22           &  \multicolumn{1}{c:}{0.9649}  & 27.03 & \multicolumn{1}{c:}{0.9783}            &120.0s   & \textcolor{black}{3.005s}  \\
         \textbf{RED (DEQ)}       & \textcolor{black}{\textbf{24.10}}           &  \multicolumn{1}{c:}{\textcolor{black}{\textbf{0.9451}}}           & \textcolor{black}{\textbf{27.41}}            & \multicolumn{1}{c:}{\textcolor{black}{\textbf{0.9660}}}   & \textcolor{black}{\textbf{27.10}} &\multicolumn{1}{c:}{\textcolor{black}{\textbf{0.9789}}}   &  \textcolor{black}{{166.9s}}            & \textcolor{black}{{4.577s}}   \\\hline
    \end{tabular*}}
    }
\label{Tab:table3}
\end{table}

\subsection{Image Reconstruction in IDT}
\label{Sec:ValidationIDT}
IDT~\cite{Ling.etal18} is a data intensive computational imaging modality that seeks to recover the spatial distribution of the complex-valued permittivity contrast of an object given a set of its intensity-only measurements. Specifically, $\Abm$ consists of a set of $b$ complex measurement operators $[\Abm_1, \dots, \Abm_b]^{\Tsf}$, where each $\Abm_i$ is a convolution corresponding to the $i$th measurement $\ybm_i$. In the simulation, we randomly extracted and cropped  400 slices of 416$\times$416 images for training, 28 images for validation and 56 images for testing from Brecahad database~\cite{aksac2019brecahad}. Following the setup in~\cite{Ling.etal18, Wu.etal2020}, we generated $b = \,$500 intensity measurements under AWGN corresponding to $\{15, 20, 25\}$ dB of input SNR.  ODER and RED (DEQ) were trained at the noise level corresponding to 20 dB input SNR. In our comparisons, we also included the recent SGD-Net~\cite{Liu.etal2021} method that corresponds to RED (Unfold), but uses stochastic data-consistency layers similar to ODER. SGD-Net allows for more unfolded iteration blocks by improving the usage of limited GPU memory. Both ODER and RED (DEQ) were trained using SGD, while all other methods were trained using Adam~\cite{Kingma.Ba2015}.

Fig.~\ref{Fig:IDTtraining} compares the average loss and SNR achieved by RED (DEQ) with ($b=500$) and ODER with $w\in\{100, 250\}$ during training. It took 67.49 hours to train RED (DEQ) for 180 epochs. It took 24.76 and 39.23 hours to train ODER with ($w=100$) and ($w=250$), respectively, for the same number epochs. Table~\ref{Tab:table1} provides the final SNR achieved by ODER and several baseline methods on the test data. The runtime in the table corresponds to the average inference time that excludes the model loading. ODER with ($w=100$) is around $3\times$ faster than RED (DEQ) on both GPU and CPU. Fig.~\ref{Fig:testsnr}~(\emph{left}) highlights the faster convergence of ODER compared to RED (DEQ) to the similar SNR.

ODER is memory efficient due to its online processing of measurements.  The memory considerations in IDT include the size of all the variables related to the desired image $\xbm$, the measured data ${\ybm_i}$, and the variables related to the measurement operator $\{\Abm_i\}$. ODER addresses the problem of storing and processing the measurements and the measurement operators on the GPU during end-to-end training. Table~\ref{Tab:table1} shows the total memory (GB) used by ODER and RED (DEQ) for reconstructing a $416 \times 416$ pixel permittivity image.  While RED (DEQ) requires $3.56$ GB of GPU memory in every iteration,  ODER with $w=100$ requires only $0.71$ GB, which is about $20\%$ of the full volume.

\vspace{-0.5em} 
\subsection{Image Reconstruction in Sparse-View CT}
\label{Sec:ValidationCT}
We consider simulated data obtained from the clinically realistic CT images provided by Mayo Clinic for the \emph{low dose CT grand challenge}~\cite{mccollough2016tu}. Specifically, $2070$ 2D slices of size $512 \times 512$ corresponding to 7 patients were used to train the models. The test images correspond to $55$ slices randomly selected from another patient. We implement the measurement operator $\bm{A}$ and its adjoint $\Abm^\Tsf$ with PyTorch implementation of \texttt{Radon} and \texttt{IRadon}~\footnote{The code is publicly available at \url{https://github.com/phernst/pytorch_radon}} transform. We assume that the CT machine is designed to project from nominal angles with $b\in\{90, 120, 180\}$ projection views that are evenly-distributed on a half circle and 724 detector pixels. We add Gaussian noise to the sinograms to make the input SNR equal to 50 dB. We empirically found that using Adam~\cite{Kingma.Ba2015} is around $2\times$ faster than applying SGD when training both ODER and RED (DEQ). We thus trained all learning-based methods using Adam. Table~\ref{Tab:table2} reports the average SNR and SSIM results for ODER with ($w/b$) of $\{30/90,40/120,50/180\}$ and all baselines. Fig.~\ref{Fig:testsnr} (middle) reports the convergence speed of ODER with ($w=50$) for sparse-view CT with full batch ($b=180$) views. The visual comparisons are in Fig.~\ref{Fig:vsMRICT} (bottom). Note how ODER matches the performance of RED (DEQ) and outperforms RED (Denoising) and RED (Unfold) across different projection views.

\vspace{-.5em}   
\subsection{Image Reconstruction in Accelerated Parallel MRI}
\label{Sec:ValidationMRI}
We simulated a multi-coil CS-MRI setup using radial Fourier sampling~\cite{Lustig.etal2007, Lustig.etal2008}. The measurement operator $\Abm$ thus consists of a set of $b$ complex measurement operators depending on a set of receiver coils $\{\Sbm_i\}$~\cite{Pruessmann.etal1999}. For each coil, we have $\Abm_i=\Pbm\Fbm\Sbm_i$, where $\Pbm$ is the diagonal sampling matrix, $\Fbm$ is the Fourier transform, and $\Sbm_i$ is the diagonal matrix of sensitivity maps. ODER is evaluated on two brain MRI datasets. The first dataset~\cite{zhang2018ista} provides $800$ slices of $256\times256$ images for training and $50$ slices for testing. The second dataset~\cite{knoll2020fastmri} contains a randomly selected $400$ volumes of $320\times320\times10$ images for training, and $32$ volumes for testing. We synthesized the total number of ($b=96$) 2D/3D coil sensitivity maps using the $\texttt{SigPy}$~\cite{Ong.etal2019} for each dataset, respectively. Since all the CNNs in our numerical study are 2D, we apply them slice-by-slice when forming 3D volumes (all slices are passed in parallel using batch processing). We trained all learning-based methods using Adam. Fig.~\ref{Fig:testsnr} (right) reports the convergence speed of ODER for CS-MRI at $20\%$ sampling. Table~\ref{Tab:table3} reports the average SNR and SSIM values for ODER with ($w=48$) and all baseline methods. The visual comparison can be found in Fig.~\ref{Fig:vsMRICT} (top) at $10\%$ sampling.
\section{Conclusion and Future Work}
\label{sec:conclusion}
\vspace{-.8em}   
This work proposes ODER as a new online DEQ learning method for RED, analyzes its theoretical properties in terms of convergence and accuracy, and applies it to three widely-used imaging inverse problems. ODER extends the recent DEQ approach in~\cite{Gilton.etal2021} by introducing randomized processing of measurements. Our extensive theoretical and numerical results corroborate the potential of ODER to reduce the computational/memory complexity of training and testing, while achieving similar imaging quality as RED (DEQ). The future work can explore to further improve our analysis and design distributed variants of ODER to enhance its performance on parallel computing architectures.

%

\section*{Acknowledgements}
Research presented in this article was supported by the NSF CAREER award CCF-2043134.

\newpage
\appendix

\setcounter{theorem}{0}

\section*{Supplementary Material}

We adopt the monotone operator theory~\cite{Bauschke.Combettes2017, Ryu.Boyd2016} for a unified analysis of ODER. The contributions of this work are algorithmic, theoretical, and numerical. We propose ODER as a new algorithm. We then develop new theoretical insights into its convergence and ability to approximate the traditional DEQ. In Supplement A, we prove the convergence of forward pass for ODER to $\xbmbar\in\Fix(\Tsf)$ up to an error term controlled by $\gamma$ and $w$. In Supplement B, we prove that the online backward pass in expectation converges to $\bbmbar\in\Fix(\Fsf)$ up to an error term that can be controlled via $w$. In Supplement C, we prove the ability of ODER to approximate the stationary points of the desired loss $\ell(\thetabm)$ up to an error term that can be controlled during training. Finally, in Section~\ref{Sec:TechnicalDetails}, we provide additional technical details on our implementations and simulations omitted from the main paper due to space.

We use the same notations as in the main manuscript. The measurement model corresponds to $\ybm = \Abm\xbmast + \ebm$, where $\xbmast$ is the true solution and $\ebm$ is the noise. The function $g(\xbm)$ denotes the data-fidelity term. The operator $\Dsf_{\thetabm}(\cdot)$ denotes the learned prior within ODER and RED (DEQ), which is implemented via its residual $\Rsf_{\thetabm} \defn \Isf - \Dsf_{\thetabm}$. The operator $\Tsfhat_{\thetabm}(\cdot)$ and $\Fsfhat(\cdot)$  denote the ODER stochastic forward and backward passes, respectively. The operator $\Tsf_{\thetabm}(\cdot)$ and $\Fsf(\cdot)$ denote the full batch forward and backward passes of RED (DEQ), respectively. Finally, our code, including pre-trained CNN models used in ODER and RED (DEQ), is also included in the supplementary material. 

\section{Proof of Theorem 1}

The following theorem shows the convergence (in expectation) of the ODER forward pass for convex $g$ and contractive $\Dsf_\thetabm$. Note that this proof is a variation of existing results in the literature on online PnP/RED~\cite{Sun.etal2018a, Wu.etal2020, Tang.Davies2020, Sun.etal2021}. However, this result plays an important role in the analysis of both the backward pass and the ability of ODER to approximate the traditional DEQ learning.

\label{Sup:Sec:Theorem1}
\begin{theorem}
\label{Thm:ContractOnRED}
Run the forward pass of ODER for $k \geq 1$ iterations under Assumptions~1-4 using the step size $0 < \gamma < 1/(\lambda+\tau)$. Then, the sequence of forward pass iterates satisfies
\begin{equation}
\E\left[\|\xbm^k - \xbmbar\|_2\right] \leq  \eta^k R + \frac{\gamma \nu}{(1-\eta)\sqrt{w}},
\end{equation}
for some constant $0 < \eta < 1$ where $\xbmbar \in \Fix(\Tsf)$.
\end{theorem}
\begin{proof}
For notation connivance, we abbreviate $\Tsf_{\thetabm}(\cdots)$ as $\Tsf_{\thetabm}(\cdots)$ in the following proof. From Lemma~\ref{Sup:Lem:ContractionofU}, $\Tsf$ is a contraction, which means that there exists $0< \eta < 1$ such that
\begin{equation*}
\|\Tsf(\zbm)-\Tsf(\ybm)\|_2 \leq \eta \|\zbm-\ybm\|_2\,,
\end{equation*}
for all $\zbm, \ybm\in\R^n$. Then, for $\xbmbar\in\Fix(\Tsf)$, we have that 
\begin{align*}
\|\xbm^k - \xbmbar\|_2
&= \|\Tsfhat(\xbm^{k-1}) - \Tsf(\xbmbar)\|_2^2 
= \|\Tsf(\xbm^{k-1}) - \Tsf(\xbmbar) + \Tsfhat(\xbm^{k-1}) - \Tsf(\xbm^{k-1})\|_2 \\
&\leq \|\Tsf(\xbm^{k-1}) - \Tsf(\xbmbar)\|_2 + \gamma \|\Tsf(\xbm^{k-1})-\Tsfhat(\xbm^{k-1})\|_2\\
&\leq \eta \|\xbm^{k-1}-\xbmbar\|_2 + \|\nabla g(\xbm^{k-1}) - \nabla \ghat(\xbm^{k-1})\|_2,
\end{align*}
where we used the triangular inequality and that $\Tsf$ is $\eta$-Lipschitz continuous. We take the conditional expectation from both sides to obtain
\begin{equation*}
\E\left[\|\xbm^k-\xbmbar\|_2 \,|\, \xbm^{k-1}\right] \leq \eta \|\xbm^{k-1}-\xbmbar\|_2 + \frac{\gamma\nu}{\sqrt{w}},
\end{equation*}
where we applied the Jensen's inequality to the variance bound in Assumption~4. By taking the total expectation, we thus have
\begin{equation*}\E\left[\|\xbm^k-\xbmbar\|_2\right] \leq \eta \E\left[\|\xbm^{k-1}-\xbmbar\|_2\right] + \frac{\gamma\nu}{\sqrt{w}}.
\end{equation*}
By iterating this inequality and using the bound in Assumption~3, we get the result
\begin{equation*}
\E\left[\|\xbm^k-\xbmbar\|_2\right] \leq \eta^k R + \frac{\gamma \nu}{(1-\eta)\sqrt{w}}.
\end{equation*}
\end{proof}

\subsection{Useful Results for the Proof of Theorem~1}
\label{Sup:Sec:ProofLem1}

\medskip\noindent
The following lemma establishes that $\Tsf$ is a contraction. The proof is a minor modification of the Theorem 1 from~\cite{Gilton.etal2021}, which we provide for completeness. It is worth noting that this result does not assume that the functions $\{g_i\}$ are strongly convex.
\begin{figure}[t]
\centering\includegraphics[width=\textwidth]{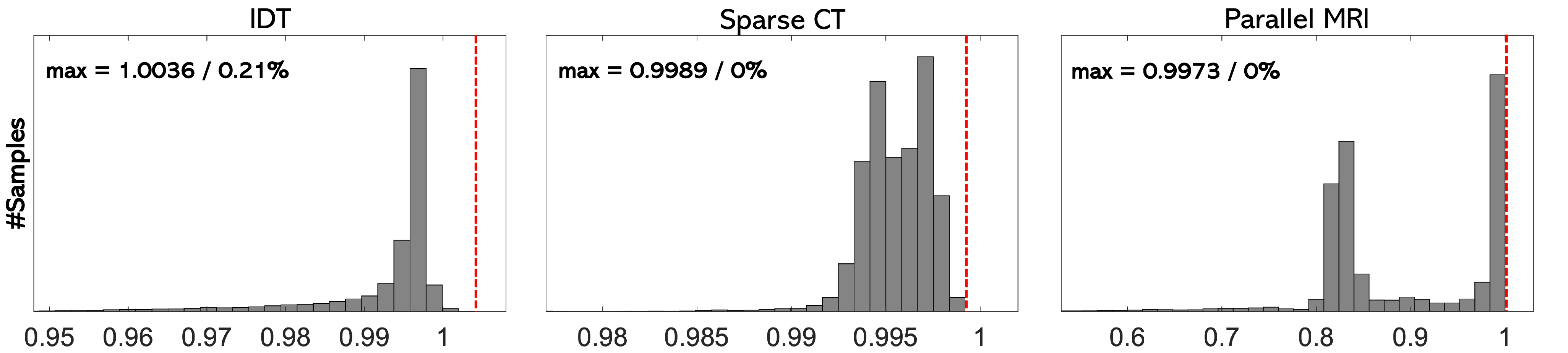}
\caption{~\emph{Empirical evaluation of the Lipschitz continuity of $\,\Tsf$. Each histogram was generated by storing all ODER iterates and $\xbmbar \in \Fix(\Tsf)$ across all the test images used in the tables of the main paper. The x-axis is the value of $\,\|\Tsf(\xbm^{k-1})-\Tsf(\xbmbar)\|_2/\|\xbm^{k-1}-\xbmbar\|_2$.~\textbf{Left:} The histogram of IDT at $b=500$ with \{15, 20, 25\} dB of input SNR. \textbf{Middle:} The histogram of sparse CT at $b\in\{90,120,180\}$ projection views. \textbf{Right:} The histogram of the radially sub-sampled parallel MRI at $10\%$ and $20\%$ sampling. Note how $\Tsf$ numerically acts as a contraction on all the iterates generated for the CT and MRI experiments, and on 99.79\% iterates generated for the IDT experiments. Despite their imperfect numerical precision, current spectral normalization techniques still provide a powerful tool for systematically ensuring stability of PnP/RED fixed-point iterations.}}
\label{Fig:etaEstimate}
\vspace{-.5em}
\end{figure}

\begin{lemma}
\label{Sup:Lem:ContractionofU}
Suppose that Assumptions 1-2 in the main paper are true. Then, for any $0 < \gamma < 1/(\lambda+\tau)$, the operator $\Tsf$ in eq.~(2) of the main paper is a contraction, which means that for all $\xbm \in \R^n$
\begin{equation*}
\|\nabla_\xbm \Tsf(\xbm)\|_2 < 1,
\end{equation*}
where $\| \cdot \|_2$ denotes the spectral norm. 
\end{lemma}
\begin{proof}
The Jacobian of the operator $\Tsf$ with respect to $\xbm$ is given by
\begin{equation*}
\nabla_\xbm \Tsf(\xbm) = (1-\gamma\tau) \Ibm - \gamma \Hsf g(\xbm) - \gamma \tau \nabla_\xbm \Dsf(\xbm).
\end{equation*}
Let $\lambda_1 \geq \cdots \geq \lambda_n$ denote sorted eigenvalues of the Hessian matrix $\Hsf g(\xbm)$. Since $g$ is convex, we have that $\lambda_n \geq 0$. Then, for any $\xbm \in \R^n$, we have
\begin{align*}
\|\nabla_\xbm \Tsf(\xbm)\|_2 
& = \|(1-\gamma\tau) \Ibm - \gamma \Hsf g(\xbm) - \gamma \tau \nabla_\xbm \Dsf(\xbm)\|_2\\
&\leq \|(1-\gamma\tau) \Ibm - \gamma \Hsf g(\xbm)\|_2 + \gamma \tau \|\nabla_\xbm \Dsf(\xbm)\|_2\\
&\leq \max_{1 \leq i \leq n} \left\{1-\gamma\tau - \gamma \lambda_i\right\} + \gamma \tau \kappa\\
&\leq 1 - \gamma \tau(1-\kappa) < 1,
\end{align*}
where in the first inequality we used the triangular inequality, in the second the fact that $\Dsf$ is a contraction, and in the third the convexity of $g$.
\end{proof}

\section{Proof of Theorem 2}
\label{Sup:Sec:Theorem 2}

The following result is a novel analysis of the ODER backward pass. The result implies that the backward pass converges (in expectation) up to an error term that can be controlled by the minibatch parameter $w$. Our numerical results provide additional corroboration to our theory by showing that ODER nearly matches the performance of the traditional DEQ learning.

\begin{theorem}
\label{Thm:NeumannOnRED}
Run the backward pass of ODER for $k \geq 1$ iterations under Assumptions~1-4 from $\bbm^0 = \zerobm$ using the step-size $0 < \gamma < 1/(\lambda+\tau)$. Then, the sequence of backward pass iterates satisfies
\begin{equation}
\E\left[\|\bbm^k-\bbmbar\|_2\right] \leq B_1 \eta^k + \frac{B_2}{\sqrt{w}},
\end{equation}
where $0 < \eta < 1$, $B_1 > 0$ and $B_2 > 0$ are constants independent of $k$ and $w$, and $\bbmbar \in \Fix(\Fsf)$.
\end{theorem}
\begin{proof}
Let $\xbm^K$ denote the output of the forward pass of ODER after $K \geq 1$ iterations, $\xbmast$ denote the training label, and $\xbmbar \in \Fix(\Tsf)$. Consider the following two operators
\begin{equation*}
\Fsf(\bbm) = [\nabla_\xbm \Tsf(\xbmbar)]^\Tsf\bbm + (\xbmbar-\xbmast) \quad\text{and}\quad\Fsfhat(\bbm) = [\nabla_\xbm \Tsfhat(\xbm^K)]^\Tsf\bbm + (\xbm^K-\xbmast),
\end{equation*}
where the first operator is used in the backward pass of RED (DEQ), while the second is its online approximation. Note also the following two Jacobians
\begin{equation*}
\nabla_\xbm \Tsf(\xbmbar) = \Ibm - \gamma (\Hsf g(\xbmbar) + \tau \nabla_\xbm \Rsf(\xbmbar))\quad\text{and}\quad \nabla_\xbm \Tsfhat(\xbm^K) = \Ibm - \gamma (\Hsf \ghat(\xbm^K) + \tau \nabla_\xbm \Rsf(\xbm^K)).
\end{equation*}
Lemma~\ref{Sup:Lem:ContractionofU} implies that $\Tsf$ is a contraction. Let $0 < \eta < 1$ denote the Lipschitz constant of $\Tsf$. Since $\nabla_\bbm \Fsf(\bbm) = \nabla_\xbm \Tsf(\xbmbar)$, we have $\|\nabla_\bbm \Fsf(\bbm)\|_2 = \|\nabla_\xbm \Tsf(\xbmbar)\|_2 \leq \eta$, which means that $\Fsf$ is a contraction
\begin{equation*}
\|\Fsf(\zbm)-\Fsf(\ybm)\|_2 \leq \eta \|\zbm-\ybm\|_2, \quad \zbm, \ybm \in \R^n.
\end{equation*}
We can thus show the following bound
\begin{align*}
\|\bbm^k-\bbmbar\|_2 &= \|\Fsfhat(\bbm^{k-1})-\Fsf(\bbmbar)\|_2 = \|\Fsf(\bbm^{k-1})-\Fsf(\bbmbar)+\Fsfhat(\bbm^{k-1})-\Fsf(\bbm^{k-1})\|_2 \\
&\leq \|\Fsf(\bbm^{k-1})-\Fsf(\bbmbar)\|_2+\|\Fsfhat(\bbm^{k-1})-\Fsf(\bbm^{k-1})\|_2 \\
&\leq \eta \|\bbm^{k-1}-\bbmbar\|_2 + \|\Fsfhat(\bbm^{k-1})-\Fsf(\bbm^{k-1})\|_2,
\end{align*}
where we first used the triangular inequality and then the fact that $\Fsf$ is a contraction. By taking the conditional expectation on both sides, we obtain
\begin{equation}
\label{Eq:Thm2:FhatF}
\E\left[\|\bbm^k-\bbmbar\|_2 \,|\, \xbm^K, \bbm^{k-1}\right] \leq \eta \|\bbm^{k-1}-\bbmbar\|_2 + \E\left[\|\Fsfhat(\bbm^{k-1})-\Fsf(\bbm^{k-1})\|_2 \,|\, \xbm^K, \bbm^{k-1} \right].
\end{equation}
We can bound the second term in~\eqref{Eq:Thm2:FhatF} as follows
\begin{align*}
&\E\left[\|\Fsfhat(\bbm^{k-1})-\Fsf(\bbm^{k-1})\|_2 \,|\, \xbm^K, \bbm^{k-1} \right] \leq \gamma \E\left[\|\Hsf g(\xbmbar)-\Hsf \ghat(\xbm^K)\|_2 \,|\, \xbm^K, \bbm^{k-1} \right] \|\bbm^{k-1}\|_2 \\
&+ \gamma \tau \|\nabla_\xbm \Rsf(\xbmbar)-\nabla_\xbm \Rsf(\xbm^K)\|_2 \|\bbm^{k-1}\|_2 + \|\xbm^K-\xbmbar\|_2 \\
&\leq \gamma \lambda \|\xbm^K-\xbmbar\|_2\|\bbm^{k-1}\|_2 + \frac{\gamma \nu}{\sqrt{w}}\|\bbm^{k-1}\|_2 + \gamma\tau\alpha \|\xbm^K-\xbmbar\|_2\|\bbm^{k-1}\|_2 + \|\xbm^K-\xbmbar\|_2 \\
&\leq A_1 \|\xbm^K-\xbmbar\|_2 + \frac{A_2}{\sqrt{w}},
\end{align*}
with $A_1 \defn (2 \gamma \lambda  + 2 \gamma\tau\alpha R + 1)$ and $A_2 = 2\nu R$, where in the second inequality we used Lemma~\ref{Lem:BoundingRandomHessianVariance} and the $\alpha$-Lipschitz continuity of $\nabla_\xbm \Rsf$ and in the third $\|\bbm^{k-1}\|_2 \leq 2R$. Since $\bbm^0 = \zerobm$, Assumption~3 implies that $\|\bbmbar\|_2 \leq R$, which leads to $\|\bbm^{k-1}\|_2 \leq 2R$ for all $k \geq 1$.

\noindent
By including the last bound into~\eqref{Eq:Thm2:FhatF}, we obtain
\begin{equation*}
\E\left[\|\bbm^k-\bbmbar\|_2 \,|\, \xbm^K, \bbm^{k-1}\right] \leq \eta \|\bbm^{k-1}-\bbmbar\|_2 + A_1 \|\xbm^K-\xbmbar\|_2 + A_2/\sqrt{w}.
\end{equation*}
By taking the total expectation and using Theorem~\ref{Thm:ContractOnRED}, we get
\begin{align*}
\E\left[\|\bbm^k-\bbmbar\|_2\right] 
\leq \eta \E\left[\|\bbm^{k-1}-\bbmbar\|_2\right] + A_1 R\eta^K + \frac{A_1\gamma \nu}{(1-\nu)\sqrt{w}} + \frac{A_2}{\sqrt{w}}.
\end{align*}
By iterating this bound and noting that $k \leq K$, we get the final result
\begin{equation*}
\E\left[\|\bbm^k-\bbmbar\|_2\right]  \leq \eta^k B_1 + \frac{B_2}{\sqrt{w}}\,,
\end{equation*}
where $B_1 \defn R + A_1R/(1-\eta)$ and $B_2 \defn (((A_1\gamma \nu)/(1-\eta)) + A_2)/(1-\nu) $.

\end{proof}

\subsection{Technical Lemma for the Proof of Theorem~2}
\label{Sup:Sec:ProofLem2}

The following technical result is used in the proof of Theorem 2. It bounds the variance of the Hessian of the data-fidelity term $g$.

\begin{lemma}
\label{Lem:BoundingRandomHessianVariance}
Suppose that Assumptions 1 and 4 in the main paper are true. Then, for any $\zbm, \ybm \in \R^n$
\begin{equation*}
\E\left[\|\Hsf g(\zbm) - \Hsf \ghat(\ybm)\|_2\right] \leq \lambda \|\zbm-\ybm\|_2 + \frac{\nu}{\sqrt{w}},
\end{equation*}
where the expectation is taken over the indices $\{i_1, \dots, i_w\}$ used for $\ghat$.
\end{lemma}

\begin{proof}
The proof directly follows the $\lambda$-Lipschitz continuity assumption of $\Hsf g$ is Assumption 1 and boundedness of the variance in Assumption 4
\begin{align*}
\E\left[\|\Hsf g(\zbm) - \Hsf \ghat(\ybm)\|_2\right] 
&\leq \E\left[\|\Hsf g(\zbm) - \Hsf g(\ybm)\|_2\right] + \E\left[\|\Hsf g(\ybm) - \Hsf \ghat(\ybm)\|_2\right] \\
&\leq \lambda \|\zbm-\ybm\|_2 + \frac{\nu}{\sqrt{w}},
\end{align*}
where we used the Jensen's inequality to get the second term.
\end{proof}

\section{Proof of Theorem 3}
\label{Sup:Sec:Theorem3}

Our final theoretical result is a novel analysis on the ability of ODER to approximate the stationary points of the loss $\ell(\thetabm)$. We show that ODER can approximate (in expectation) the stationary points up to an error term that can be controlled by the minibatch size $w$ and the learning rate $\beta$.

\begin{theorem}
\label{Thm:MainConvRes}
Train ODER using SGD for $T \geq 1$ iterations under Assumptions~1-6 using the step-size parameters $0 < \beta \leq 1/L$ and the minibatch size $w \geq 1$. Select a large enough number of forward and backward pass iterations $K \geq 1$ to satisfy $0<\eta^K\leq1/{\sqrt{w}}$. Then, we have that
\begin{equation*}
\label{EqTheorem1}
\frac{1}{T}\sum_{t=0}^{T-1}\E\left[\|\nabla\ell(\thetabm^t)\|_2^2\right]\leq\frac{2(\ell(\thetabm^0)-\ell(\thetabm^{\ast}))}{\beta T} + \frac{C_1}{\sqrt{w}} + \beta C_2\,.
\end{equation*}
where $C_1>0$ and $C_2>0$ are constants independent of $T$ and $w$.
\end{theorem}
\begin{proof}
\noindent
Consider the RED (DEQ) loss $\ell$ and its ODER approximation $\ellhat$
\begin{equation}
\label{Eq:UnbiasedAverages}
\ell(\thetabm) = \frac{1}{p}\sum_{j=1}^{p}\ell_j(\thetabm)\quad\text{and}\quad \ellhat(\thetabm) = \frac{1}{p}\sum_{j=1}^{p}\ellhat_j(\thetabm)\,,
\end{equation}
where each $\ell_j$ and $\ellhat_j$ have the forms
\begin{equation*}
\ell_j(\thetabm)=\frac{1}{2}\|\xbmbar_j(\thetabm)-\xbm_j^{\ast}\|_2^2\quad\text{and}\quad\ellhat_j(\thetabm)=\frac{1}{2}\|\xbm^K_j(\thetabm)-\xbm_j^{\ast}\|_2^2.
\end{equation*}
Vector $\xbm_j^K$ denotes the final iterate of the online forward pass obtained after $K \geq 1$ iterations for the training sample $j\in\{1,\cdots,p\}$. 

\medskip\noindent
From Assumption~5, we obtain the traditional Lipschitz continuity bound on the gradient
\begin{equation*}
\|\nabla \ell(\thetabm_1) -\nabla \ell(\thetabm_2) \|_2 \leq L\|\thetabm_1 - \thetabm_2\|_2,
\end{equation*}
which directly leads to  traditional quadratic upper bound (see Lemma 1.2.3 in~\cite{Nesterov2004}).
\begin{align}
\label{Eq:QuadBound}
\ell(\thetabm_1) \leq \ell(\thetabm_2) &+ \nabla \ell(\thetabm_2)^{\Tsf}(\thetabm_1 - \thetabm_2) + \frac{ L}{2}\|\thetabm_1-\thetabm_2\|_2^2.
\end{align}
Lemma~\ref{Sup:Lem:Boundedlosshat} in this supplement establishes the following useful bound for our analysis
\begin{equation*}
\E\left[\|\nabla\ellhat(\thetabm^t)-\nabla\ell(\thetabm^t)\|_2^2\right]\leq\frac{C}{\sqrt{w}},
\end{equation*}
for some constant $C > 0$ (see its full expression in Lemma~\ref{Sup:Lem:Boundedlosshat}). This directly implies that
\begin{align}
\label{Eq:LosshatBound}
\E\left[-\nabla\ell(\thetabm^t)^{\Tsf}\nabla\ellhat(\thetabm^t) + \frac{1}{2}\|\nabla\ellhat(\thetabm^t)\|_2^2\right]\leq-\frac{1}{2}\E[\|\nabla\ell(\thetabm^t)\|_2^2] + \frac{C}{2\sqrt{w}}\,.
\end{align}

\noindent
Consider a single iteration of SGD for optimizing ODER
\begin{equation}
\thetabm^{t+1}=\thetabm^t - \beta \nabla\ellhat_{j_t}(\thetabm^t).
\end{equation}
From the quadratic upper bound~\eqref{Eq:QuadBound}, we get
\begin{equation}
\begin{aligned}
\label{Eq:quadupperBound}
\nonumber\ell(\thetabm^{t+1})  - \ell(\thetabm^{t})
\nonumber&\leq\nabla \ell(\thetabm^t)^{\Tsf}(\thetabm^{t+1} - \thetabm^t) + \frac{L}{2}\|\thetabm^{t+1}-\thetabm^t\|_2^2\\
\nonumber&= -\beta\nabla \ell(\thetabm^t)^{\Tsf} \nabla \hat\ell_{j_t}(\thetabm^t)+ \frac{\beta^2 L}{2}\|\nabla\hat\ell_{j_t}(\thetabm^t)\|_2^2\,.
\end{aligned}
\end{equation}
For notation convenience, we use $\E[\cdot|\setminus j_t]$ to denote the expectation only with respect to the training index $j_t \in\{1,\dots,p\}$, where we condition on $\thetabm^t$ and the random indices within forward and backward passes at this iteration. By taking $\E[\cdot|\setminus j_t]$ on both sides of the quadratic upper bound~\eqref{Eq:QuadBound}
\begin{align}
\label{Eq: conditionBound}
\E\left[\ell(\thetabm^{t+1})|\setminus j_t \right] -\ell(\thetabm^t)
\nonumber&\leq-\beta\nabla \ell(\thetabm^t)^{\Tsf} \E\left[\nabla\hat\ell_{j_t}(\thetabm^t)|\setminus j_t\right] +  \frac{\beta^2 L}{2}\E\left[\|\nabla\ellhat_{j_t}(\thetabm^t)\|_2^2|\setminus j_t\right]\\
&=-\beta\nabla \ell(\thetabm^t)^{\Tsf}\nabla\hat\ell(\thetabm^t) +  \frac{\beta^2 L}{2}\E\left[\|\nabla\ellhat_{j_t}(\thetabm^t)\|_2^2|\setminus j_t\right]\,,
\end{align}
where we use~\eqref{Eq:UnbiasedAverages}  and the fact that $j_t$ is distributed uniformly at random in $\{1, \dots, p\}$.

\noindent
We now estimate the last term in~\eqref{Eq: conditionBound}
\begin{equation}
\begin{aligned}
\label{Eq:sto2stoclossBound}
&\E\left[\|\nabla\ellhat_{j_t}(\thetabm^t)\|_2^2|\setminus j_t\right]\\
&=\E\left[\|\nabla\ellhat_{j_t}(\thetabm^t)-\nabla\ellhat(\thetabm^t) + \nabla\ellhat(\thetabm^t)\|_2^2|\setminus j_t\right]\\
&=\E\left[\|\nabla\ellhat_{j_t}(\thetabm^t)-\nabla\ellhat(\thetabm^t)\|_2^2 + 2(\nabla\ellhat_{j_t}(\thetabm^t)-\nabla\ellhat(\thetabm^t))^{\Tsf}\nabla\ellhat(\thetabm^t) +\|\nabla\ellhat(\thetabm^t)\|_2^2 | \setminus j_t\right]\\
&=\E\left[\|\nabla\ellhat_{j_t}(\thetabm^t)-\nabla\ellhat(\thetabm^t)\|_2^2 | \setminus j_t\right] + \E\left[\|\nabla\ellhat(\thetabm^t)\|_2^2| \setminus j_t\right]\,,
\end{aligned}
\end{equation}
where in the third equality we use the fact that $\E\left[(\nabla\ellhat_{j_t}(\thetabm^t)-\nabla\ellhat(\thetabm^t))^{\Tsf}\nabla\ellhat(\thetabm^t)|\setminus j_t\right]=0$.%

\noindent
By replacing the last term in~\eqref{Eq: conditionBound} with~\eqref{Eq:sto2stoclossBound} and taking the full expectation on both sides of~\eqref{Eq: conditionBound} in terms of all random variables and invoking the bound~\eqref{Eq:LosshatBound}, we obtain
\begin{equation*}
\begin{aligned}
&\E[\ell(\thetabm^{t+1})] - \E\left[\ell(\thetabm^t)\right]\\
&\leq\E\left[-\beta\nabla \ell(\thetabm^t)^{\Tsf}\nabla\hat\ell(\thetabm^t) + \frac{\beta^2L}{2}\|\nabla\ellhat(\thetabm^t)\|\right] + \frac{\beta^2L}{2}\E\left[\|\nabla\ellhat_{j_t}(\thetabm^t)-\nabla\ellhat(\thetabm^t)\|_2^2\right]\\
&\leq -\frac{\beta}{2}\E[\|\nabla\ell(\thetabm^t)\|_2^2] + \frac{\beta C}{2\sqrt{w}} + \frac{\beta^2 L}{2}(4\sigma^2+\frac{6C}{\sqrt{w}})\,,
\end{aligned}
\end{equation*}

\noindent
where we used the fact that $0<\beta\leq 1/L$, and in the last inequality we used the bound from Lemma~\ref{Sup:Lem:Boundedlosshat2} in this document. By rearranging the terms, and summing this bound over $0\leq t \leq T-1$, we obtain
\begin{equation*}
\begin{aligned}
\frac{1}{T}\sum_{t=0}^{T-1}\E\left[\|\nabla\ell(\thetabm^t)\|_2^2\right]
&\leq\frac{2(\ell(\thetabm^0)-\E[\ell(\thetabm^{T})])}{\beta T} + \frac{C_1}{\sqrt{w}} + \beta C_2\\
&\leq\frac{2(\ell(\thetabm^0)-\ell(\thetabm^{\ast}))}{\beta T} + \frac{C_1}{\sqrt{w}} + \beta C_2\,,
\end{aligned}
\end{equation*}
where $C_1\defn7\gamma^2\tau^2\alpha^2(R+R^2)(B_1+B_2+R^2 +  \gamma\nu R/(1-\eta))$ and $C_2\defn4L\sigma^2$ are constants independent of $t$ and $w$. In the last inequality we used the fact that $\ell(\thetabm^{\ast})\leq\ell(\thetabm)$ for all $\thetabm$, where $\thetabm^{\ast}$ denotes a global minimizer of $\ell$.
\end{proof}

\subsection{Technical Lemmas for the Proof of Theorem~3}
\label{Sup:Sec:ProofLem1-2}
The following lemma are useful for relating $\nabla\ellhat$ and $\nabla\ell$ in expectation up to an error term. Both are used in the proof of Theorem~\ref{Thm:MainConvRes}.
\begin{lemma}
\label{Sup:Lem:Boundedlosshat}
Given the loss function $\ell$ of RED (DEQ) and $\hat\ell$ of ODER, by selecting a large enough number of forward and backward iterations $K\geq1$ to satisfy $0<\eta^K\leq1/{\sqrt{w}}$, we have 
\begin{equation*}
\E\left[\|\nabla\ellhat(\thetabm)-\nabla\ell(\thetabm)\|_2^2\right]\leq \frac{C}{\sqrt{w}}, 
\end{equation*}
where $C\defn\gamma^2\tau^2\alpha^2(R+R^2)(B_1+B_2+R^2 +  \gamma\nu R/(1-\eta))$ is a constant.
\end{lemma}
\begin{proof}
By using the DEQ training, we have the following bound
\begin{align}
\label{Eq:stcerrorlossBound}
\|\nabla\ellhat_{j}(\thetabm) - \nabla\ell_{j}(\thetabm)\|_2
\nonumber&= \left\|\left[\nabla_\thetabm \Tsfhat_{\thetabm}({\xbm_{j}^K})\right]^{\Tsf}\bbm_{j}^K - \left[\nabla_\thetabm \Tsf_{\thetabm}({\xbmbar_{j}})\right]^{\Tsf}\bbmbar_{j}\right\|_2\\
\nonumber&= \left\|\left[\nabla_\thetabm \Tsfhat_{\thetabm}({\xbm_{j}^K})\right]^{\Tsf}(\bbm_j^K -\bbmbar_j) + \left[\nabla_\thetabm \Tsfhat_{\thetabm}({\xbm_{j}^K}) -  \nabla_\thetabm \Tsf_{\thetabm}({\xbmbar_{j}})\right]^{\Tsf}\bbmbar_j\right\|_2\\
\nonumber&\leq\left\|\nabla_\thetabm \Tsfhat_{\thetabm}({\xbm_{j}^K})\right\|_2 \|\bbm_j^K - \bbmbar_j\|_2 + \left\|\nabla_\thetabm \Tsfhat_{\thetabm}({\xbm_{j}^K}) -  \nabla_\thetabm \Tsf_{\thetabm}({\xbmbar_{j}}) \right\|_2\|\bbmbar_j\|_2\\
&\leq \gamma\tau\alpha(\|\bbm_j^K - \bbmbar_j\|_2 + \|\xbm_j^K - \xbmbar_j\|_2\|\bbmbar_j\|_2),\quad\forall j\in\{1,\cdots,p\}\,,
\end{align}
where in the first inequality, we used Cauchy-Schwarz inequality and the fact that  $\Dsf_{\thetabm}$ and $\nabla_{\thetabm}\Dsf_{\thetabm}(\xbm)$ are $\alpha$-Lipschitz continuous with respect to $\xbm$ and $\thetabm$ based on Assumption~2. By applying Assumption 3 which states that $\|\bbm^k-\bbmbar\|_2<R$ and $\|\xbm^k-\xbmbar\|_2<R$ for every $\xbm, \bbm\in\R^n$, we have 
\begin{equation}
\label{Lem:GlobalBoundStochGrad}
\|\nabla\ellhat_{j}(\thetabm) - \nabla\ell_{j}(\thetabm)\|_2\leq\gamma\tau\alpha(R+R^2),\quad\forall j\in\{1,\cdots,p\}
\end{equation}
\noindent
On the other hand, by taking the expectation with respect to the stochastic approximation variables in forward and backward propagation in~\eqref{Eq:stcerrorlossBound}, invoking the bounds obtained from Theorem~1 and Theorem~2 and using the fact that $0<\eta^K\leq1/\sqrt{w}$, we have
\begin{equation*}
\E\left[\|\nabla\ellhat_{j}(\thetabm) - \nabla\ell_{j}(\thetabm)\|_2\right]\leq \gamma\tau\alpha\left(\frac{B_1+B_2+R^2}{\sqrt{w}}+\frac{\gamma\nu R}{(1-\eta)\sqrt{w}}\right)\,
\end{equation*}
where $B_1>0$ and $B_2>0$ are constants obtained in Theorem~\ref{Thm:NeumannOnRED} .
\noindent
Now consider a random variable $X$ which has a probability density function given by  $f(x)$ on the real number line such that $P(0\leq X\leq c) =1$, then we have
\begin{equation*}
\E[X] = \int_0^c x^2f(x)dx\leq\int_0^c cxf(x)dx=c\E[X].
\end{equation*}
As a consequence, given the bound~\eqref{Lem:GlobalBoundStochGrad} for the random variables $\|\nabla\ellhat_{j}(\thetabm) - \nabla\ell_{j}(\thetabm)\|_2$ and using the above fact, we have the following useful bound for our proof
\begin{equation}
\label{Eq:stostohat2sto}
\E\left[\|\nabla\ellhat_{j}(\thetabm) - \nabla\ell_{j}(\thetabm)\|_2^2\right]\leq\frac{C}{\sqrt{w}}\,,
\end{equation}
where $C\defn\gamma^2\tau^2\alpha^2(R+R^2)(B_1+B_2+R^2 +  \gamma\nu R/(1-\eta))$ is a constant.

\noindent
Note also the fact that for $\abm_1,\cdots,\abm_p\in\R^n$, we have
\begin{equation*}
\left\|\sum_{j=1}^{p}\abm_j\right\|_2^2\leq p\sum_{j=1}^{p}\|\abm_j\|_2^2.
\end{equation*}
As a result, applying the bound~\eqref{Eq:stostohat2sto}, we have
\begin{equation*}
\begin{aligned}
\E\left[\|\nabla\ellhat(\thetabm) - \nabla\ell(\thetabm)\|_2^2\right]
&=\frac{1}{p^2}\E\left[\left\|\sum_{j=1}^p(\nabla\ellhat_j(\thetabm) - \nabla\ell_j(\thetabm))\right\|_2^2\right]\\
&\leq\frac{1}{p}\sum_{j=1}^p\E\left[\left\|\nabla\ellhat_j(\thetabm) - \nabla\ell_j(\thetabm)\right\|_2^2\right]\\
&\leq\frac{C}{\sqrt{w}}\,.
\end{aligned}
\end{equation*}
This finishes the proof.
\end{proof}
\medskip\noindent
\begin{lemma}
\label{Sup:Lem:Boundedlosshat2}
Given the loss function $\ell$ and $\hat\ell$, by taking the conditional expectation with respect to the training index $j_t$ (via conditioning on all other variables), we have 
\begin{equation*}
\E\left[\|\nabla\ellhat_{j_t}(\thetabm^t)-\nabla\ellhat(\thetabm^t)\|_2^2\right]\leq 4\sigma^2+\frac{6C}{\sqrt{w}}, \end{equation*}
where $C>0$ obtained in~Lemma~\ref{Sup:Lem:Boundedlosshat} is a constant.
\end{lemma}
\begin{proof}
By taking expectation with respect to training index $j_t\in\{1\cdots,N\}$, we obtain the following useful bound for our proof
\begin{equation}
\begin{aligned}
&\E\left[\|\nabla\ell_{j_t}(\thetabm^t)-\nabla\ellhat(\thetabm^t)\|_2^2|\setminus j_t\right]\\
&=\E\left[\|\nabla\ell_{j_t}(\thetabm^t) -\nabla\ell(\thetabm^t) + \nabla\ell(\thetabm^t)- \nabla\ellhat(\thetabm^t)\|_2^2|\setminus j_t\right]\\
&=\E\left[\|\nabla\ell_{j_t}(\thetabm^t) -\nabla\ell(\thetabm^t)\|_2^2 + \|\nabla\ell(\thetabm^t)- \nabla\ellhat(\thetabm^t)\|_2^2\right.\\
&\quad\quad + \left.2(\nabla\ell_{j_t}(\thetabm^t) -\nabla\ell(\thetabm^t))^{\Tsf}(\nabla\ell(\thetabm^t)- \nabla\ellhat(\thetabm^t))|\setminus j_t\right]\\
&\leq 2\left(\E\left[\|\nabla\ell_{j_t}(\thetabm^t) -\nabla\ell(\thetabm^t)\|_2^2 + \|\nabla\ell(\thetabm^t)- \nabla\ellhat(\thetabm^t)\|_2^2|\setminus j_t\right]\right)\\
&=2\E\left[\|\nabla\ell_{j_t}(\thetabm^t) -\nabla\ell(\thetabm^t)\|_2^2|\setminus j_t\right] + 2\|\nabla\ell(\thetabm^t)- \nabla\ellhat(\thetabm^t)\|_2^2\,,
\end{aligned}
\end{equation}
where we used Young's inequality that states for any $\abm_1,\abm_2\in\R^n$, we have
\begin{equation*}
2\abm_1^{\Tsf}\abm_2 \leq \|\abm_1\|_2^2 + \|\abm_2\|_2^2\quad\Rightarrow\quad\|\abm_1+\abm_2\|_2^2\leq2(\|\abm_1\|_2^2 + \|\abm_2\|_2^2).
\end{equation*}
\noindent
By taking the full expectation of the inequality (13) above and applying Lemma~\ref{Sup:Lem:Boundedlosshat} and the bounded variance in Assumption 6, we have 
\begin{equation}
\label{Eq:StoFullLFullStoLBound}
\begin{aligned}
&\E\left[\|\nabla\ell_{j_t}(\thetabm^t)-\nabla\ellhat(\thetabm^t)\|_2^2\right]\\
&\leq 2\E\left[\E\left[\|\nabla\ell_{j_t}(\thetabm^t)-\nabla\ell(\thetabm^t)\|_2^2|\setminus j_t\right]\right] +2\E\left[\|\nabla\ellhat(\thetabm^t) - \nabla\ell(\thetabm^t)\|_2^2\right]\\
&\leq 2\sigma^2 + \frac{2C}{\sqrt{w}}
\end{aligned}
\end{equation}

\noindent
Similarly, by taking full expectation and using Lemma~\ref{Sup:Lem:Boundedlosshat} and the bound~\eqref{Eq:StoFullLFullStoLBound}, we write that 
\begin{equation*}
\label{Eq:StoStoLFullStoLBound}
\begin{aligned}
&\E\left[\|\nabla\ellhat_{j_t}(\thetabm^t)-\nabla\ellhat(\thetabm^t)\|_2^2\right]\\
&=\E\left\|\nabla\ellhat_{j_t}(\thetabm^t) -\nabla\ell_{j_t}(\thetabm^t) + \nabla\ell_{j_t}(\thetabm^t)- \nabla\ellhat(\thetabm^t)\|_2^2\right]\\
&\leq2\E\left[\|\nabla\ell_{j_t}(\thetabm^t)- \nabla\ellhat(\thetabm^t)\|_2^2\right] + 2\E\left[\|\nabla\ellhat_{j_t}(\thetabm^t) -\nabla\ell_{j_t}(\thetabm^t)\|_2^2\right] \\
&\leq4\sigma^2+\frac{6C}{\sqrt{w}}.
\end{aligned}
\end{equation*}
\end{proof}
\section{Additional Technical Details and Numerical Results}
\label{Sec:TechnicalDetails}

In this section, we present technical details that were omitted from the main paper for space. We used the following \emph{signal-to-noise ratio (SNR)}~\cite{Gupta2018, Wu.etal2020} in dB for quantitively comparing different algorithms
\begin{equation}
\label{Eq:SNR}
\hbox{SNR}(\xbfhat,\xbm) = \max_{a,b\in \R} \left\{20 \log_{10}\left( \frac{\|\xbm\|_2}{\|\xbm - a\xbfhat + b\|_2}\right)  \right\},
\end{equation}
where $\xbfhat$ and $\xbm$ represents the noisy vector and ground truth respectively, while the purpose of $a$ and $b$ is to adjust for contrast and offset. We also used the \emph{structural similarity index measure (SSIM)}~\cite{Wang.etal2004} as an alternative metric. All the experiments in this work were performed on a machine equipped with an Intel Xeon Gold 6130 Processor and eight NVIDIA GeForce RTX 2080 Ti GPUs.

As stated in~\cite{Bai.etal2019, Gilton.etal2021} and other DEQ work, using acceleration can reduce computational costs during both training and inference time and lead to improvement of empirical performance at inference. Here, we focus on the final image reconstruction performance for denoising based step-descent RED (SD-RED) by using two different fixed-point acceleration methods, namely Anderson acceleration and Nesterov acceleration. The detailed instructions of using Anderson acceleration is publicly available with tutorials~\footnote{Anderson acceleration for DEQ was introduced at \url{http://implicit-layers-tutorial.org/.}}. The Nesterov acceleration for RED (DEQ) and RED (Denoising) can be summarized as
\begin{equation*}
\begin{aligned}
\label{Eq:PnP-Netupdate}
&\xbm^k = \Tsf_{\thetabm}(\sbm^k)\\
&c_k = (q_{k-1} -1 )/q_k\\
&\sbm^k = \xbm^k + c_k(\xbm^k - \xbm^{k-1}) \;,
\end{aligned}
\end{equation*}
where the value of $q_k = 1/2(1+\sqrt{1+4q_{k-1}^2})$ is adapted for better PSNR performance. The average SNR (dB) values for RED (Anderson) and RED (Nesterov) using different CNN denoisers on the MRI images are presented in Table~\ref{Tab:table2}. We empirically observe that the RED with Nesterov acceleration led to better reconstructions in terms of SNR. For the forward pass iterations, we equip ODER, RED (Denoising), RED (Unfold) and RED (DEQ)
\begin{table}[t]
\caption{Average SNR (dB) for different pre-trained CNNs on MRI test images. Note that the ``AWGN denoising'' performance is for noise level $\sigma=5$ and the ``Time (ms)'' presents the runtime of evaluating $\Rsf_{\thetabm}(\xbm)/\nabla_{\xbm}\Rsf_{\thetabm}(\xbm)$ on images of size $320\times320$.}
    \centering
    \renewcommand\arraystretch{1.2}
    {\footnotesize
    \scalebox{1}{
    \begin{tabular*}{12.5cm}{L{88pt}||C{88pt}lC{88pt}lC{88pt}|}
        \hline
         \diagbox{\bf SNR(dB)}{\bf Model} & \textbf{DnCNN} & \textbf{Tiny U-Net} & \textbf{U-Net}\\ \hline\hline
        
         \textbf{AWGN denoising}     & \multicolumn{1}{c}{30.30}  &  \multicolumn{1}{c}{30.36} & \multicolumn{1}{c}{30.41} \\
         \textbf{RED (Nesterov)}       & \multicolumn{1}{c}{26.37}  &   \multicolumn{1}{c}{26.35} & \multicolumn{1}{c}{26.42}\\
         \textbf{RED (Anderson)}      & \multicolumn{1}{c}{25.44}  & \multicolumn{1}{c}{25.46}  & \multicolumn{1}{c}{25.51} \\
         \textbf{Time (ms)}     & \multicolumn{1}{c}{12.22 / 31.74}  &  \multicolumn{1}{c}{1.65 / 11.48} & \multicolumn{1}{c}{1.88 / 32.84} \\\hline
    \end{tabular*}}
    }
\label{Tab:table2}
\end{table}
with Nesterov acceleration for all experiments used in this work. We utilize Anderson acceleration for the backward pass for both ODER and RED (DEQ). We limit the number of backward pass iterations to $50$ for efficiency considerations for all three imaging applications. The number of forward passes is presented in each imaging modality sub-section, respectively.  Followed by~\cite{Gilton.etal2021, bai.etal2022}, we additionally set the convergence criterion (relative norm difference between iterations) as
\begin{equation*}
\frac{\|\xbm^{k+1} - \xbm^{k}\|_2}{\|\xbm^{k}\|_2} < \epsilon\,,
\end{equation*}
where $\epsilon > 0$. In forward passes, We set $\epsilon=10^{-3}$ for ODER and RED (DEQ), while we set stopping criterion of backward passes to $\epsilon=10^{-2}$ for ODER and RED (DEQ).

We additionally tested three network architectures including DnCNN~\cite{Zhang.etal2017}, U-Net~\cite{Ronneberger.etal2015} and tiny U-Net~\cite{Liu.etal2021}. The DnCNN network has seventeen layers, including 15 hidden layers, an input layer, and an output layer. The tiny U-Net is a simplified variant of the normal U-Net with less trainable parameters. In specific, the CNN consists of four scales, each with a skip connection between downsampling and upsampling. These connections increase the effective receptive field of the CNN. The number of channels in each layer are \{32, 64, 128, 256\}. We make two additional modifications to the tiny U-Net. First, we drop out the second group normalization (GN)~\cite{Wu.etal2018} at each composite convolutional layers. Second, we add spectral normalization to each layers for more stable training and better Lipschitz constrain of the neural network. It is worth to note that spectral normalization is a widely used method for Lipschitz constrained neural network, and it is ~\emph{not} our aim to claim any algorithmic novelty with respect to it. In Table~\ref{Tab:table2}, we present the denoising performance on AWGN removal with noise level $\sigma=5$ and the run time of calculating $\Rsf_{\thetabm}(\xbm)/\nabla_{\xbm}\Rsf_{\thetabm}(\xbm)$ with respect to a $320\times320$ image. Overall, the traditional U-Net architecture achieves the best denoising and reconstruction performance, but requires more time per iterate than tiny U-Net. As a result, we implement traditional U-Net denoiser for RED (Denoising), and we equip ODER, RED (DEQ) and RED (Unfold) for the same tiny U-Net architecture in order to decrease per-iteration computation costs during training.
\begin{figure}[t]
\centering\includegraphics[width=1.01\textwidth]{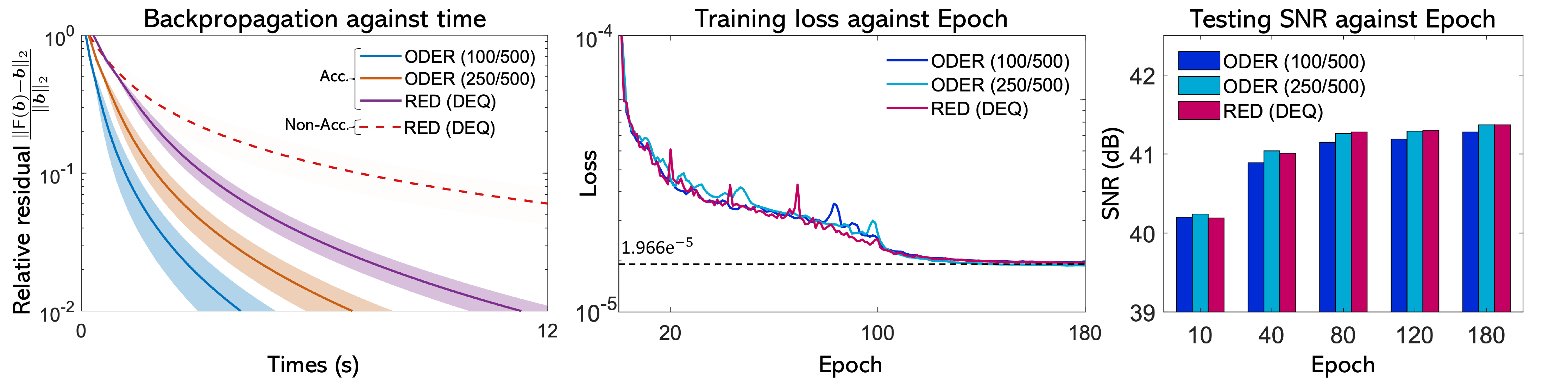}
\caption{~\emph{Numerical Illustration of ODER on IDT for two minibatch sizes $w \in \{100, 250\}$. The result for RED (DEQ) with $b = 500$ is also provided for reference. The left figure shows how ODER improves the efficiency of the backward pass of RED (DEQ) by reducing the per-iteration complexity of the measurement matrix. The middle figure plots the loss against the epoch number evaluated on the training set. The right figure plots the SNR (dB) achieved at different epochs for ODER evaluated over the testing set. This figure highlights that by using minibatches $1\leq w\leq b$, ODER improves per-iteration complexity and matches the same final imaging quality achieved by RED (DEQ).}}
\label{Fig:IDTtraining}
\vspace{-.5em}
\end{figure}

\begin{table}[t]
  \caption{Average SSIM values for IDT image recovery on testing images from~\cite{aksac2019brecahad}.}
      \centering
      \renewcommand\arraystretch{1.2}
      {\footnotesize
      \scalebox{1}{
      \begin{tabular*}{11.90cm}{L{80pt}||C{90pt}lC{90pt}lC{90pt}lC{90pt}l}
          \hline
          \multirow{2}{4em}{\textbf{Method}}& \multicolumn{3}{c}{\textbf{Input SNR (dB)}}\\
          \cline{2-4}
          &\multicolumn{1}{c}{15}  & \multicolumn{1}{c}{20} & \multicolumn{1}{c}{25} \\\hline\hline
          \textbf{TV}       & {0.9810}    & {0.9829}   & \multicolumn{1}{c}{0.9835} \\
          \textbf{U-Net}       & {0.9811}      & {0.9831}       & \multicolumn{1}{c}{0.9836}  \\        
           \textbf{ISTA-Net+}       & 0.9809   & 0.9833    & \multicolumn{1}{c}{0.9841}  \\
            \textbf{SGD-Net (100)}    & \textcolor{black}{0.9832}     & {0.9859}     & \multicolumn{1}{c}{\textcolor{black}{0.9866}} \\
           \textbf{RED (Denoising)}      & {0.9831}   & {0.9852}     & \multicolumn{1}{c}{0.9866} \\
           \cdashline{1-9}
           \textbf{ODER (100)}      & 0.9834    & 0.9875     &  \multicolumn{1}{c}{0.9889} \\
           \textbf{ODER (250)}      &  \textcolor{black}{\textbf{0.9846}}    & 0.9878 &  \textcolor{black}{\textbf{0.9891}} \\
           \textbf{RED (DEQ)}       & 0.9845  & \textcolor{black}{\textbf{0.9879} }    &  \multicolumn{1}{c}{0.9890} \\\hline
      \end{tabular*}}
      }
  \label{Tab:table1}
  \end{table}

In Fig.~\ref{Fig:etaEstimate}, we report the empirical evaluation of the Lipschitz constant $\eta$ of $\Tsf$ used in our simulations and stated in Lemma~\ref{Sup:Lem:ContractionofU} on the testing images from all three inverse problems in the main manuscript. We plot the histograms of values $\eta=\|\Tsf(\xbm^{k-1})-\Tsf(\xbmbar)\|_2/\|\xbm^{k-1}-\xbmbar\|_2$, and the maximum value of each histogram is indicated by a vertical bar with the frequency of $\eta>1$, providing an empirical upper bound on the values of $\eta$. Note that despite the numerical limitations of current spectral normalization techniques, they still provide a useful tool to ensure stable convergence.

\subsection{Additional Details and Validations for IDT}
\label{Sup:Sec:SuppIDT}

We follow the experimental setup in~\cite{Wu.etal2020, Sun.etal2021, Ling.etal18} to generate the measurement matrix and simulated images for IDT~\footnote{The code is publicly available at \url{https://github.com/bu-cisl/High-Throughput-IDT}}.  In specific, the simulated images are assumed to be on the focal plane $z\,=\,0\SI{}{\micro\metre}$ with LEDs located at $z_{\text{LED}} = -70 \SI{}{\milli\metre}$. The wavelength of the illumination was set to $\lambda = 630\SI{}{\nano\metre}$ and the background medium index was assumed to be water with $\epsilon_b = 1.33$.  We generated $b = 500\,$ intensity measurements with $40\times$ microscope objectives (MO) and 0.65 numerical aperture (NA). Followed by~\cite{Sun.etal2021}, we assume real permittivity function, and our implementation stores each $\Abm_i$ as two separate arrays for phase and absorption. In addition, each matrix is stored in the Fourier space to reduce the computational complexity of evaluating convolutions. This result in the storage of complex valued arrays for each, consisting of pairs of single precision floats for every element when training ODER/RED (DEQ). Thus, the shape of each measurements and measurement operators in full batch RED (DEQ) for reconstructing one slice is $1 \times 416\times 416\times 500 \times2$. A detailed discussion on the IDT forward model is available in~\cite{Ling.etal18, Wu.etal2020}. 

We train ODER and RED (DEQ) with the initialization $\xbm_0=\Abm^{\Hsf}\ybm$, where $\Abm^{\Hsf}$ denotes the conjugate transpose. For both ODER and RED (DEQ) during training, we fix the step-size parameter and regularization parameter to $\gamma=5\times10^{-3}$ and $\tau=4$, respectively. The learning rate of ODER/RED (DEQ) is set in two stages. In the first 100 epochs, we adopt the cyclic learning rate policy~\cite{Smith2017}, where the policy cycles the learning rate between $0.05$ and $0.16$ with exponentially decay to $0.9998$. In stage 2, the learning rate was gradually reduced by a factor of $0.6$ every $50$ epochs. The number of total training epochs was $200$.
We set the same forward pass initialization $\xbm^0$ in ODER for all reference methods. In these experiments, we set the number of forward pass iterations in ODER/RED (DEQ) to $K=80$, and we set the steps in RED (Denoising) and RED (Unfold) to $K=100$ and $K=9$, respectively. 

Table~\ref{Tab:table1} reports average SSIM values obtained by ODER and other baselines. Fig.~\ref{Fig:IDTtraining} 
presents quantitative evaluation of ODER on IDT for two minibatch sizes $w\in\{100,250\}$ against RED (DEQ) using the full-batch of $b=500$ measurements. Specifically, Fig.~\ref{Fig:IDTtraining} (\emph{left}) presents the empirical acceleration of ODER backward pass over that of RED (DEQ) due to the reduction in the computation complexity of data-consistency blocks at each iteration. Fig.~\ref{Fig:IDTtraining} (\emph{middle}) illustrates the loss against the epoch number on the training set, while Fig.~\ref{Fig:IDTtraining} (\emph{right}) presents the SNR achieved at different epoches for different values of $w$ evaluated over testing set. Fig.~\ref{Fig:IDTvisual} provides additional visualizations of IDT reconstruction produced by ODER and other reference methods.
\begin{figure}[t]
\centering\includegraphics[width=0.9\textwidth]{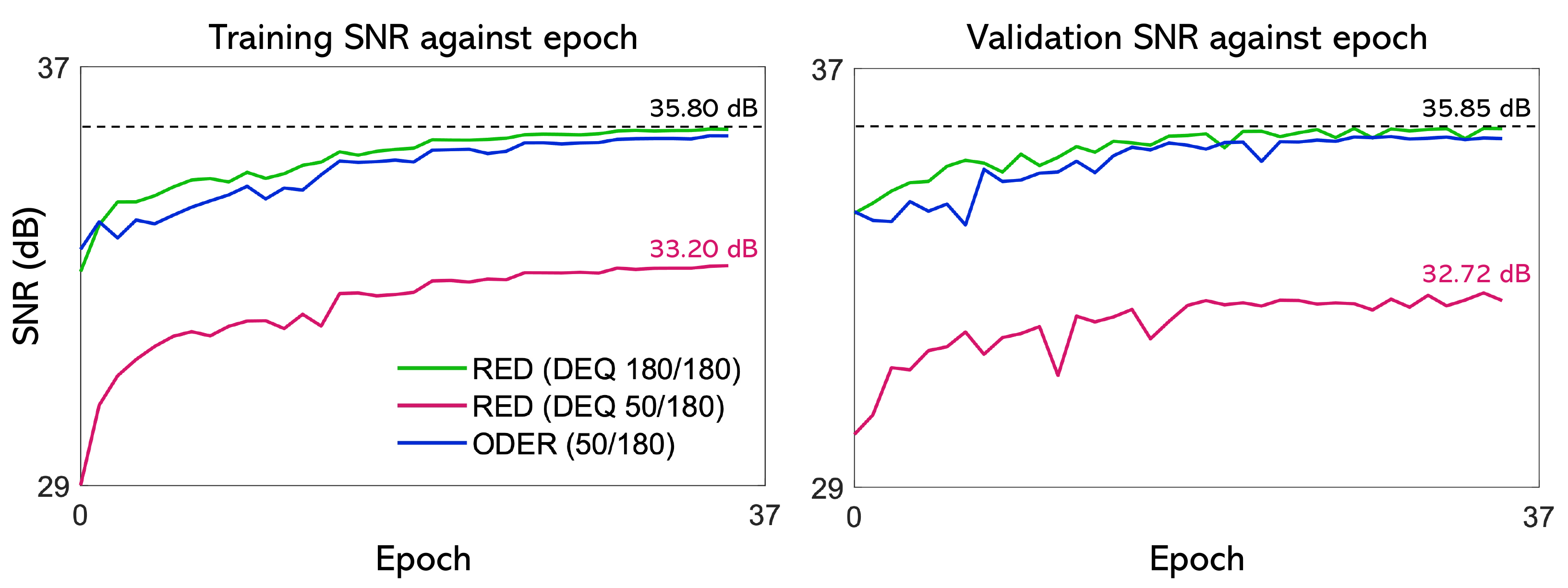}
\caption{~\emph{Average SNR of ODER and RED (DEQ) on the taining and validation images of sparse-view CT are plotted against the training epoch for a fixed per-iteration measurement budget. The total number of projection views is $b=180$. Under the same per-iteration memory complexity, the training of ODER achieves a significantly higher SNR than RED (DEQ) when the budget is limited to $50$ projection views at per-iterate due to its ability to randomly cycle through the full dataset.}}
\label{Fig:memorybudget}
\vspace{-.5em}
\end{figure}

\subsection{Additional Details and Validations for CT/MRI}
\textbf{Sparse-view CT.} For the CT images, we train ODER by using the filtered backprojection (FBP) initialization $\xbm^0 = \Abm^{\Hsf}F\ybm$. We use the Hann filter for FBP reconstruction. We set the number of forward pass steps in ODER/RED (DEQ) to $K=180$, and we use Adam with training minibatch size $4$ and weight decay $1\times 10^{-7}$. For ODER and RED (DEQ), we fixed the step-size to $\gamma=1.25\times 10^{-3}$ and regularization parameter to $\tau=3$. The learning rate starts from $3\times 10^{-4}$ and is halved at epoch 15, then gradually reduced by a factor of 0.6 every 5 epochs. The number of total training epochs is 35. It is worth to note that we equally divided the full projection views $b\in\{90, 120, 180\}$ into $5$ non-overlapping chunks, each with size of $\{18, 24, 36\}$ views, respectively. At every iterations, the corresponding ODER model with $w/b\in\{30/90,40/120,50/180\}$ randomly picks a subset of $\{6, 8, 10\}$ from each chunk for the data-consistency block calculation. This leads to better empirical reconstruction improvement for ODER. Similarly, we set the number of forward pass iterations in RED (Denoising) and RED (Unfold) to $K=150$ and $K=7$, respectively. 

Fig.~\ref{Fig:memorybudget} compares the average reconstruction SNR of ODER and RED (DEQ) for a fixed periteration measurement budget on sperase-view CT. The total number of projection views is $b=180$. The batch algorithm RED (DEQ 50/180) and ODER with ($w=50$) are allowed to use only 50 projection views at per iterate. This means that in each figure both RED (DEQ 50/180) and ODER with ($w=50$) have the same per-iteration computational complexity. Empirically, we observe that ODER outperforms RED (DEQ 50/180) by around 3 dB on the validation set under the same per-step memory complexity and matches the performance of RED (DEQ 180/180). In Fig.~\ref{Fig:Visualctmri} (\emph{top}), we provide additional visualizations of the solutions produced by ODER/RED (DEQ) and various baseline methods considered in our work.

\vspace{+0.3em}
\noindent
\textbf{Parallel MRI.} For the MRI images, we train ODER/RED (DEQ) by using the zero-filled reconstruction as initialization $\xbm^0$. For the 2D MRI images, we set the number of iterations in ODER/RED (DEQ) to $K=200$, and we use Adam with training minibatch size $16$ and weight decay $2\times 10^{-7}$.  We fixed the step-size to $\gamma=1.2$ and regularization parameter to $\tau=0.05$ for both ODER and RED (DEQ). The learning rate starts from $1\times 10^{-4}$ and is gradually reduced by a factor of 0.6 every 10 epochs. The number of total training epochs is 40. 
In these experiments, we set the number of forward pass iterations in RED (Denoising) and RED (Unfold) to $K=200$ and $K=25$, respectively. For the 3D MRI volumes, we set the number of steps in ODER/RED (DEQ) to $K=400$, and we use Adam with training minibatch size $4$ and weight decay $1\times 10^{-6}$.  We fixed the step-size to $\gamma=1.25$ and regularization parameter to $\tau=0.01$ for both ODER and RED (DEQ). The learning rate starts from $5\times 10^{-5}$ and is gradually reduced by a factor of 0.65 every 10 epochs. The number of total training epochs is 100. Similar to sparse-view CT, we equally divided the number of coil sensitive maps $b=96$ into $4$ non-overlapping chunks, each with size of $24$ coils, respectively. At every iterations, the corresponding ODER model with $w=48$ randomly picks a subset of $12$ from each chunk for the data-consistency block calculation. For the 3D MRI volumes, we set the number of forward pass iterations in RED (Denoising) and RED (Unfold) to $K=400$ and $K=20$, respectively. Fig.~\ref{Fig:Visualctmri} (\emph{bottom}) reports the comparison on medical brain images for CS-MRI with under-sampling ratio of $10\%$. 

\begin{figure}[t]
\centering\includegraphics[width=0.95\textwidth]{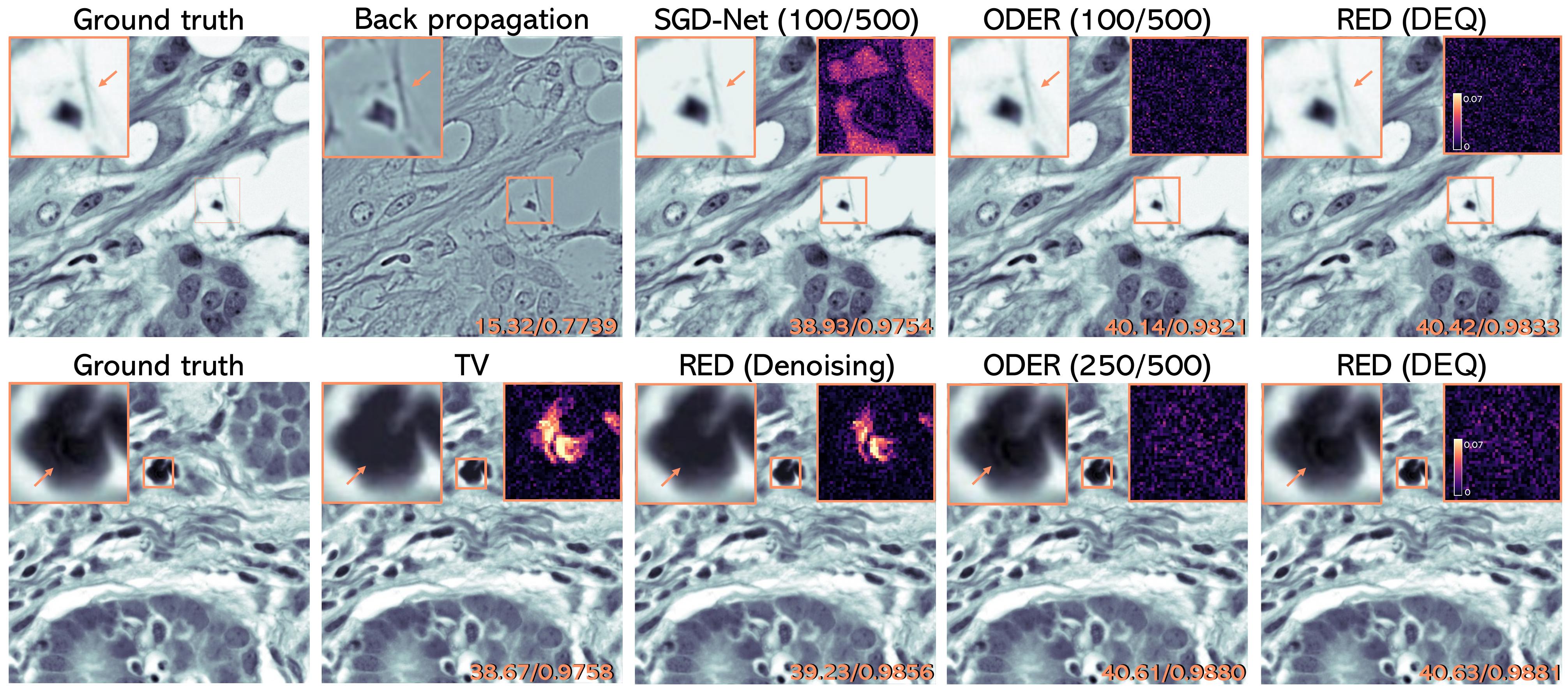}
\caption{~\emph{Quantitative evaluation of several well-known methods on IDT under noise corresponding to input SNR of 20 dB. The total number of IDT measurements in this experiment is $b = 500$. RED (DEQ) corresponds to the full batch architecture that uses all the measurements at every step. Each image is labeled with its SNR (dB) and SSIM values with respect to the original image. The yellow box provides a close-up with a corresponding error map provided on its right.  Note the similar performance of ODER and RED (DEQ), and the improvement over RED (Denoising) /RED (Unfold) due to the usage of DEQ learning. Best viewed on a digital display.}}
\label{Fig:IDTvisual}
\vspace{-.5em}
\end{figure}
\begin{figure}[t]
\centering\includegraphics[width=0.81\textwidth]{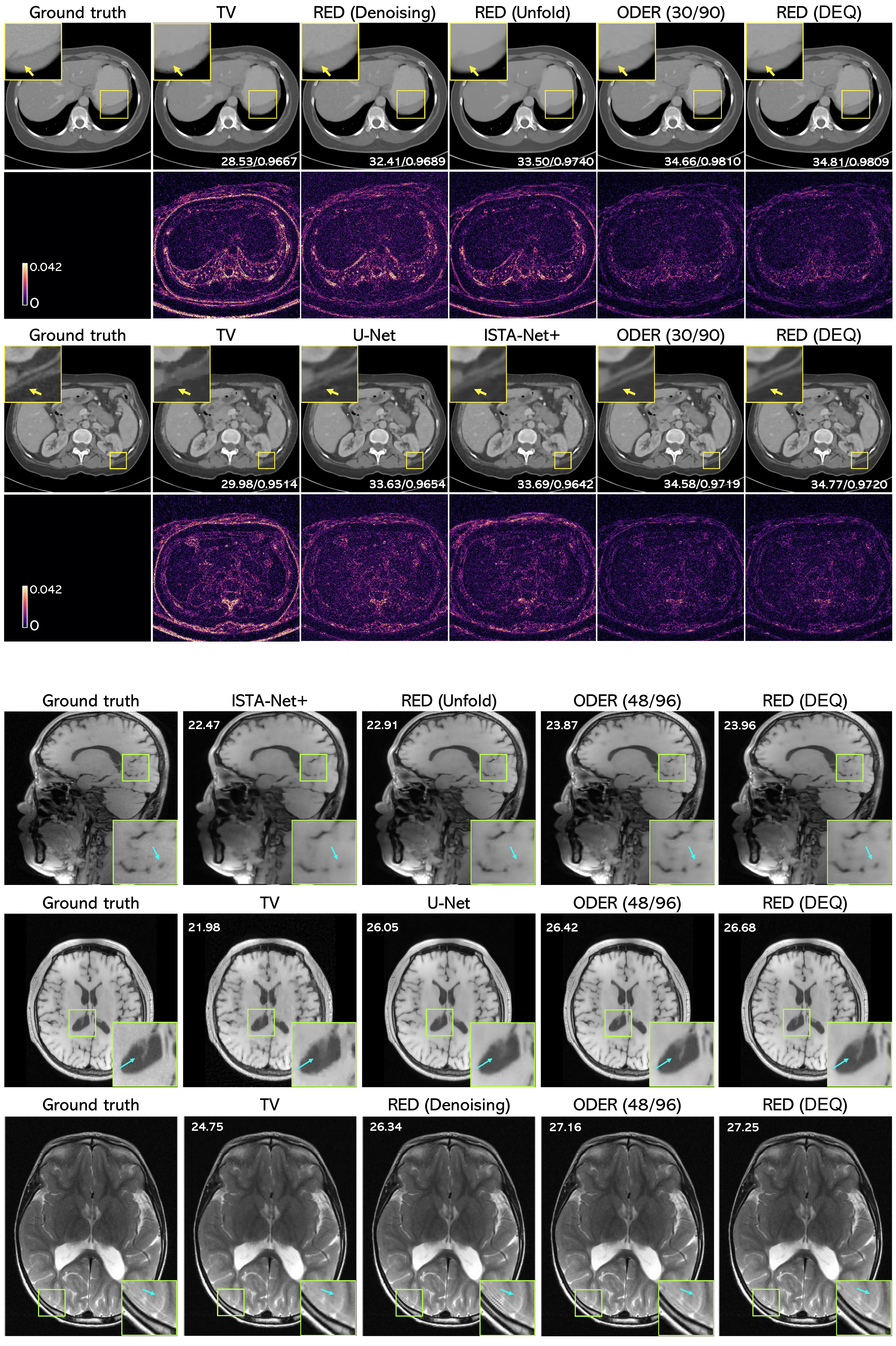}
\caption{~\emph{Visual evaluation of several well-known methods on two imaging problems: (top) Reconstruction of sparse-view CT from $b=90$ projection views. Each image is labeled with the corresponding SNR (dB) and SSIM values. The figures below are the error residual images to the ground truth; (bottom) Reconstruction of brain MRI images from its radial Fourier measurements at $10\%$ sampling with $b=96$ simulated coil sensitivity maps. Best viewed on a digital display.}}
\label{Fig:Visualctmri}
\vspace{-.5em}
\end{figure}


\begin{thebibliography}{10}
\providecommand{\url}[1]{#1}
\csname url@samestyle\endcsname
\providecommand{\newblock}{\relax}
\providecommand{\bibinfo}[2]{#2}
\providecommand{\BIBentrySTDinterwordspacing}{\spaceskip=0pt\relax}
\providecommand{\BIBentryALTinterwordstretchfactor}{4}
\providecommand{\BIBentryALTinterwordspacing}{\spaceskip=\fontdimen2\font plus
\BIBentryALTinterwordstretchfactor\fontdimen3\font minus
  \fontdimen4\font\relax}
\providecommand{\BIBforeignlanguage}[2]{{%
\expandafter\ifx\csname l@#1\endcsname\relax
\typeout{** WARNING: IEEEtran.bst: No hyphenation pattern has been}%
\typeout{** loaded for the language `#1'. Using the pattern for}%
\typeout{** the default language instead.}%
\else
\language=\csname l@#1\endcsname
\fi
#2}}
\providecommand{\BIBdecl}{\relax}
\BIBdecl

\bibitem{McCann.etal2017}
M.~T. McCann, K.~H. Jin, and M.~Unser,
\newblock ``Convolutional neural networks for inverse problems in imaging: A
  review,''
\newblock {\em IEEE Signal Process. Mag.}, vol. 34, no. 6, pp. 85--95, 2017.

\bibitem{Lucas.etal2018}
A.~Lucas, M.~Iliadis, R.~Molina, and A.~K. Katsaggelos,
\newblock ``Using deep neural networks for inverse problems in imaging:
  {B}eyond analytical methods,''
\newblock {\em IEEE Signal Process. Mag.}, vol. 35, no. 1, pp. 20--36, Jan.
  2018.

\bibitem{Ongie.etal2020}
G.~Ongie, A.~Jalal, C.~A. Metzler, R.~G. Baraniuk, A.~G. Dimakis, and
  R.~Willett,
\newblock ``Deep learning techniques for inverse problems in imaging,''
\newblock {\em IEEE J. Sel. Areas Inf. Theory}, vol. 1, no. 1, pp. 39--56, May
  2020.

\bibitem{Ronneberger.etal2015}
O.~Ronneberger, P.~Fischer, and T.~Brox,
\newblock ``{U}-{N}et: {C}onvolutional networks for biomedical image
  segmentation,''
\newblock in {\em Proc. Med. Image. Comput. Comput. Assist. Intervent.}, 2015,
  pp. 234--241.

\bibitem{DJin.etal2017}
K.~H. {Jin}, M.~T. {McCann}, E.~{Froustey}, and M.~{Unser},
\newblock ``Deep convolutional neural network for inverse problems in
  imaging,''
\newblock {\em IEEE Trans. Image Process.}, vol. 26, no. 9, pp. 4509--4522,
  Sep. 2017.

\bibitem{Kang.etal2017}
Eunhee Kang, Junhong Min, and Jong~Chul Ye,
\newblock ``A deep convolutional neural network using directional wavelets for
  low-dose x-ray {CT} reconstruction,''
\newblock {\em Medical Physics}, vol. 44, no. 10, pp. e360--e375, 2017.

\bibitem{Chen.etal2017}
H.~Chen, Y.~Zhang, M.~K. Kalra, F.~Lin, Y.~Chen, P.~Liao, J.~Zhou, and G.~Wang,
\newblock ``Low-dose {CT} with a residual encoder-decoder convolutional neural
  network,''
\newblock {\em IEEE Trans. Med. Imag.}, vol. 36, no. 12, pp. 2524--2535, Dec.
  2017.

\bibitem{Sun.etal2018}
Y.~Sun, Z.~Xia, and U.~S. Kamilov,
\newblock ``Efficient and accurate inversion of multiple scattering with deep
  learning,''
\newblock {\em Opt. Express}, vol. 26, no. 11, pp. 14678--14688, May 2018.

\bibitem{han.etal2018}
Y.~{Han} and J.~C. {Ye},
\newblock ``Framing {U-Net} via deep convolutional framelets: Application to
  sparse-view {CT},''
\newblock {\em IEEE Trans. Med. Imag.}, vol. 37, no. 6, pp. 1418--1429, 2018.

\bibitem{Venkatakrishnan.etal2013}
S.~V. Venkatakrishnan, C.~A. Bouman, and B.~Wohlberg,
\newblock ``Plug-and-play priors for model based reconstruction,''
\newblock in {\em Proc. IEEE Global Conf. Signal Process. and Inf. Process.},
  Austin, TX, USA, Dec. 3-5, 2013, pp. 945--948.

\bibitem{Romano.etal2017}
Y.~Romano, M.~Elad, and P.~Milanfar,
\newblock ``The little engine that could: {R}egularization by denoising
  ({RED}),''
\newblock {\em SIAM J. Imaging Sci.}, vol. 10, no. 4, pp. 1804--1844, 2017.

\bibitem{Zhang.etal2017}
K.~Zhang, W.~Zuo, Y.~Chen, D.~Meng, and L.~Zhang,
\newblock ``Beyond a {G}aussian denoiser: {R}esidual learning of deep {CNN} for
  image denoising,''
\newblock {\em IEEE Trans. Image Process.}, vol. 26, no. 7, pp. 3142--3155,
  Jul. 2017.

\bibitem{Zhang.etal2021b}
K.~Zhang, Y.~Li, W.~Zuo, L.~Zhang, L.~Van~Gool, and R.~Timofte,
\newblock ``Plug-and-play image restoration with deep denoiser prior,''
\newblock {\em IEEE Trans. Patt. Anal. and Machine Intell.}, pp. 1--1, 2021.

\bibitem{Chan.etal2016}
S.~H. Chan, X.~Wang, and O.~A. Elgendy,
\newblock ``Plug-and-play {ADMM} for image restoration: Fixed-point convergence
  and applications,''
\newblock {\em IEEE Trans. Comp. Imag.}, vol. 3, no. 1, pp. 84--98, Mar. 2017.

\bibitem{Sreehari.etal2016}
S.~Sreehari, S.~V. Venkatakrishnan, B.~Wohlberg, G.~T. Buzzard, L.~F. Drummy,
  J.~P. Simmons, and C.~A. Bouman,
\newblock ``Plug-and-play priors for bright field electron tomography and
  sparse interpolation,''
\newblock {\em IEEE Trans. Comput. Imaging}, vol. 2, no. 4, pp. 408--423, Dec.
  2016.

\bibitem{Kamilov.etal2017}
U.~S. Kamilov, H.~Mansour, and B.~Wohlberg,
\newblock ``A plug-and-play priors approach for solving nonlinear imaging
  inverse problems,''
\newblock {\em IEEE Signal. Proc. Let.}, vol. 24, no. 12, pp. 1872--1876, Dec.
  2017.

\bibitem{Buzzard.etal2017}
G.~T. Buzzard, S.~H. Chan, S.~Sreehari, and C.~A. Bouman,
\newblock ``Plug-and-play unplugged: {O}ptimization free reconstruction using
  consensus equilibrium,''
\newblock {\em SIAM J. Imaging Sci.}, vol. 11, no. 3, pp. 2001--2020, Sep.
  2018.

\bibitem{Reehorst.Schniter2019}
E.~T. Reehorst and P.~Schniter,
\newblock ``Regularization by denoising: Clarifications and new
  interpretations,''
\newblock {\em IEEE Trans. Comput. Imag.}, vol. 5, no. 1, pp. 52--67, Mar.
  2019.

\bibitem{Ryu.etal2019}
Ernest~K. Ryu, J.~Liu, S.~Wang, X.~Chen, Z.~Wang, and W.~Yin,
\newblock ``Plug-and-play methods provably converge with properly trained
  denoisers,''
\newblock in {\em Proc. 36th Int. Conf. Mach. Learn.}, Long Beach, CA, USA,
  Jun. 09--15 2019, vol.~97, pp. 5546--5557.

\bibitem{Mataev.etal2019}
Gary Mataev, Peyman Milanfar, and Michael Elad,
\newblock ``Deep{RED}: Deep image prior powered by {RED},''
\newblock in {\em Proc. {IEEE} Int. Conf. Comput. Vis. Workshops}, Oct. 2019,
  pp. 1--10.

\bibitem{Wu.etal2020}
Z.~{Wu}, Y.~{Sun}, A.~{Matlock}, J.~{Liu}, L.~{Tian}, and U.~S. {Kamilov},
\newblock ``{SIMBA}: Scalable inversion in optical tomography using deep
  denoising priors,''
\newblock {\em IEEE J. Sel. Topics Signal Process.}, vol. 14, no. 6, pp.
  1163--1175, Oct. 2020.

\bibitem{Liu.etal2020}
J.~{Liu}, Y.~{Sun}, C.~{Eldeniz}, W.~{Gan}, H.~{An}, and U.~S. {Kamilov},
\newblock ``{RARE}: Image reconstruction using deep priors learned without
  ground truth,''
\newblock {\em IEEE J. Sel. Topics Signal Process.}, vol. 14, no. 6, pp.
  1088--1099, Oct. 2020.

\bibitem{zhang2018ista}
J.~{Zhang} and B.~{Ghanem},
\newblock ``{ISTA-Net}: {I}nterpretable optimization-inspired deep network for
  image compressive sensing,''
\newblock in {\em Proc. {IEEE} Conf. Comput. Vision Pattern Recognit.}, 2018,
  pp. 1828--1837.

\bibitem{Hauptmann.etal2018}
A.~{Hauptmann}, F.~{Lucka}, M.~{Betcke}, N.~{Huynh}, J.~{Adler}, B.~{Cox},
  P.~{Beard}, S.~{Ourselin}, and S.~{Arridge},
\newblock ``Model-based learning for accelerated, limited-view 3-d
  photoacoustic tomography,''
\newblock {\em IEEE Trans. Med. Imag.}, vol. 37, no. 6, pp. 1382--1393, 2018.

\bibitem{Adler.etal2018}
J.~{Adler} and O.~{{\"O}ktem},
\newblock ``Learned primal-dual reconstruction,''
\newblock {\em IEEE Trans. Med. Imag.}, vol. 37, no. 6, pp. 1322--1332, June
  2018.

\bibitem{Aggarwal.etal2019}
H.~K. {Aggarwal}, M.~P. {Mani}, and M.~{Jacob},
\newblock ``Modl: Model-based deep learning architecture for inverse
  problems,''
\newblock {\em IEEE Trans. Med. Imag.}, vol. 38, no. 2, pp. 394--405, Feb.
  2019.

\bibitem{Hosseini.etal2019}
S.~A. Hosseini, B.~Yaman, S.~Moeller, M.~Hong, and M.~Akcakaya,
\newblock ``Dense recurrent neural networks for accelerated {MRI}:
  {H}istory-cognizant unrolling of optimization algorithms,''
\newblock {\em IEEE J. Sel. Topics Signal Process.}, vol. 14, no. 6, pp.
  1280--1291, Oct. 2020.

\bibitem{Yaman.etal2020}
B.~Yaman, S.~A.~H. Hosseini, S.~Moeller, J.~Ellermann, K.~U{\u g}urbil, and
  M.~Ak{\c c}akaya,
\newblock ``Self-supervised learning of physics-guided reconstruction neural
  networks without fully sampled reference data,''
\newblock {\em Magn. Reson. Med.}, vol. 84, no. 6, pp. 3172--3191, Jul. 2020.

\bibitem{Mukherjee.etal2021}
S.~Mukherjee, M.~Carioni, O.~{\"O}ktem, and C.~B. Sch{\"o}nlieb,
\newblock ``End-to-end reconstruction meets data-driven regularization for
  inverse problems,''
\newblock in {\em Advances in Neural Information Processing Systems},
  A.~Beygelzimer, Y.~Dauphin, P.~Liang, and J.~Wortman Vaughan, Eds., 2021.

\bibitem{chen.etal2018}
R.~T.~Q. Chen, Y.~Rubanova, J.~Bettencourt, and D.~K. Duvenaud,
\newblock ``Neural ordinary differential equations,''
\newblock in {\em Advances in Neural Information Processing Systems},
  S.~Bengio, H.~Wallach, H.~Larochelle, K.~Grauman, N.~Cesa-Bianchi, and
  R.~Garnett, Eds. 2018, vol.~31, Curran Associates, Inc.

\bibitem{Dupont.etal2019}
E.~Dupont, A.~Doucet, and Y.~W. Teh,
\newblock ``Augmented neural odes,''
\newblock in {\em Advances in Neural Information Processing Systems}, 2019,
  vol.~32.

\bibitem{Kelly.etal2020}
J.~Kelly, J.~Bettencourt, M.~J. Johnson, and D.~K. Duvenaud,
\newblock ``Learning differential equations that are easy to solve,''
\newblock in {\em Advances in Neural Information Processing Systems},
  H.~Larochelle, M.~Ranzato, R.~Hadsell, M.F. Balcan, and H.~Lin, Eds. 2020,
  vol.~33, pp. 4370--4380, Curran Associates, Inc.

\bibitem{Bai.etal2019}
S.~Bai, J.~Z. Kolter, and V.~Koltun,
\newblock ``Deep equilibrium models,''
\newblock {\em Proc. Advances in Neural Information Processing Systems 33},
  vol. 32, 2019.

\bibitem{Winston.etal2019}
E.~Winston and J.~Z. Kolter,
\newblock ``Monotone operator equilibrium networks,''
\newblock in {\em Advances in Neural Information Processing Systems},
  H.~Larochelle, M.~Ranzato, R.~Hadsell, M.F. Balcan, and H.~Lin, Eds. 2020,
  vol.~33, pp. 10718--10728, Curran Associates, Inc.

\bibitem{Gilton.etal2021}
D.~Gilton, G.~Ongie, and R.~Willett,
\newblock ``Deep equilibrium architectures for inverse problems in imaging,''
\newblock {\em IEEE Trans. Comput. Imag.}, vol. 7, pp. 1123--1133, 2021.

\bibitem{Kawaguchi.etal2021}
K.~Kawaguchi,
\newblock ``On the theory of implicit deep learning: Global convergence with
  implicit layers,''
\newblock in {\em International Conference on Learning Representations}, 2021.

\bibitem{Fung.etal2021}
S.~W. Fung, H.~Heaton, Q.~Li, D.~McKenzie, S.~Osher, and W.~Yin,
\newblock ``Jfb: Jacobian-free backpropagation for implicit networks,''
\newblock {\em arXiv preprint arXiv:2103.12803}, 2021.

\bibitem{Gurumurthy.etal2021}
S.~Gurumurthy, S.~Bai, Z.~Manchester, and J.~Z. Kolter,
\newblock ``Joint inference and input optimization in equilibrium networks,''
\newblock in {\em Advances in Neural Information Processing Systems},
  M.~Ranzato, A.~Beygelzimer, Y.~Dauphin, P.S. Liang, and J.~Wortman Vaughan,
  Eds. 2021, vol.~34, pp. 16818--16832, Curran Associates, Inc.

\bibitem{Bottou.Bousquet2007}
L.~Bottou and O.~Bousquet,
\newblock ``The tradeoffs of large scale learning,''
\newblock in {\em Proc. Advances Neural Inf. Process. Syst.}, Vancouver, BC,
  Canada, Dec. 3-6, 2007, pp. 161--168.

\bibitem{Bertsekas2011}
D.~P. Bertsekas,
\newblock ``Incremental proximal methods for large scale convex optimization,''
\newblock {\em Math. Program. Ser. B}, vol. 129, pp. 163--195, 2011.

\bibitem{Kim.etal2013}
D.~{Kim}, D.~{Pal}, J.~{Thibault}, and J.~A. {Fessler},
\newblock ``Accelerating ordered subsets image reconstruction for {X}-ray {CT}
  using spatially nonuniform optimization transfer,''
\newblock {\em IEEE Trans. Med. Imag.}, vol. 32, no. 11, pp. 1965--1978, Nov.
  2013.

\bibitem{Bottou.etal2018}
L.~Bottou, F.~E. Curtis, and J.~Nocedal,
\newblock ``Optimization methods for large-scale machine learning,''
\newblock {\em SIAM Rev.}, vol. 60, no. 2, pp. 223--311, 2018.

\bibitem{Ling.etal18}
R.~Ling, W.~Tahir, H.-Y. Lin, H.~Lee, and L.~Tian,
\newblock ``High-throughput intensity diffraction tomography with a
  computational microscope,''
\newblock {\em Biomed. Opt. Express}, vol. 9, no. 5, pp. 2130--2141, May 2018.

\bibitem{Kak.Slaney1988}
A.~C. Kak and M.~Slaney,
\newblock {\em Principles of Computerized Tomographic Imaging},
\newblock {IEEE}, 1988.

\bibitem{Griswold2002}
M~A Griswold, P~M Jakob, R~M Heidemann, M~Nittka, V~Jellus, J~Wang, B~Kiefer,
  and A~Haase,
\newblock ``Generalized autocalibrating partially parallel acquisitions
  (grappa),''
\newblock {\em Magn. Reson. in Med.}, vol. 47, no. 6, pp. 1202--1210, 2002.

\bibitem{Uecker.etal2014}
M.~Uecker, P.~Lai, M.~J. Murphy, P.~Virtue, M.~Elad, J.~M. Pauly, S.~S.
  Vasanawala, and M.~Lustig,
\newblock ``Espirit---an eigenvalue approach to autocalibrating parallel mri:
  where sense meets grappa,''
\newblock {\em Magnetic resonance in medicine}, vol. 71, no. 3, pp. 990--1001,
  2014.

\bibitem{Rudin.etal1992}
L.~I. Rudin, S.~Osher, and E.~Fatemi,
\newblock ``Nonlinear total variation based noise removal algorithms,''
\newblock {\em Physica D}, vol. 60, no. 1--4, pp. 259--268, Nov. 1992.

\bibitem{Bioucas-Dias.Figueiredo2007}
J.~M. Bioucas-Dias and M.~A.~T. Figueiredo,
\newblock ``A new {T}w{IST}: {T}wo-step iterative shrinkage/thresholding
  algorithms for image restoration,''
\newblock {\em IEEE Trans. Image Process.}, vol. 16, no. 12, pp. 2992--3004,
  December 2007.

\bibitem{Beck.Teboulle2009a}
A.~Beck and M.~Teboulle,
\newblock ``Fast gradient-based algorithm for constrained total variation image
  denoising and deblurring problems,''
\newblock {\em IEEE Trans. Image Process.}, vol. 18, no. 11, pp. 2419--2434,
  November 2009.

\bibitem{Dabov.etal2007}
K.~Dabov, A.~Foi, V.~Katkovnik, and K.~Egiazarian,
\newblock ``Image denoising by sparse {3-D} transform-domain collaborative
  filtering,''
\newblock {\em IEEE Trans. Image Process.}, vol. 16, no. 16, pp. 2080--2095,
  Aug. 2007.

\bibitem{Kamilov.etal2022}
U.~S. Kamilov, C.~A. Bouman, G.~T. Buzzard, and B.~Wohlberg,
\newblock ``Plug-and-play methods for integrating physical and learned models
  in computational imaging,''
\newblock 2022,
\newblock arXiv:2203.17061.

\bibitem{Zhang.etal2017a}
K.~Zhang, W.~Zuo, S.~Gu, and L.~Zhang,
\newblock ``Learning deep {CNN} denoiser prior for image restoration,''
\newblock in {\em Proc. {IEEE} Conf. Computer Vision and Pattern Recognition
  ({CVPR})}, Honolulu, USA, July 21-26, 2017, pp. 3929--3938.

\bibitem{Metzler.etal2018}
C.~Metzler, P.~Schniter, A.~Veeraraghavan, and R.~Baraniuk,
\newblock ``pr{D}eep: Robust phase retrieval with a flexible deep network,''
\newblock in {\em Proc. 36th Int. Conf. Mach. Learn.}, Stockholmsm{\"a}ssan,
  Stockholm Sweden, Jul. 10--15 2018, pp. 3501--3510.

\bibitem{Dong.etal2019}
W.~{Dong}, P.~{Wang}, W.~{Yin}, G.~{Shi}, F.~{Wu}, and X.~{Lu},
\newblock ``Denoising prior driven deep neural network for image restoration,''
\newblock {\em IEEE Trans. Pattern Anal. Mach. Intell.}, vol. 41, no. 10, pp.
  2305--2318, Oct 2019.

\bibitem{Zhang.etal2019}
K.~Zhang, W.~Zuo, and L.~Zhang,
\newblock ``Deep plug-and-play super-resolution for arbitrary blur kernels,''
\newblock in {\em Proc. {IEEE} Conf. Computer Vision and Pattern Recognition
  ({CVPR})}, Long Beach, CA, USA, June 16-20, 2019, pp. 1671--1681.

\bibitem{Sun.etal2019b}
Y.~{Sun}, S.~{Xu}, Y.~{Li}, L.~{Tian}, B.~{Wohlberg}, and U.~S. {Kamilov},
\newblock ``Regularized {F}ourier ptychography using an online plug-and-play
  algorithm,''
\newblock in {\em Proc. {IEEE} Int. Conf. Acoustics, Speech and Signal Process.
  ({ICASSP})}, Brighton, UK, May 12-17, 2019, pp. 7665--7669.

\bibitem{Ahmad.etal2020}
R.~{Ahmad}, C.~A. {Bouman}, G.~T. {Buzzard}, S.~{Chan}, S.~{Liu}, E.~T.
  {Reehorst}, and P.~{Schniter},
\newblock ``Plug-and-play methods for magnetic resonance imaging: Using
  denoisers for image recovery,''
\newblock {\em IEEE Signal Processing Magazine}, vol. 37, no. 1, pp. 105--116,
  2020.

\bibitem{Wei.etal2020}
K.~Wei, A.~Aviles-Rivero, J.~Liang, Y.~Fu, C.-B. Sch\"onlieb, and H.~Huang,
\newblock ``Tuning-free plug-and-play proximal algorithm for inverse imaging
  problems,''
\newblock in {\em Proc. 37th Int. Conf. Machine Learning ({ICML})}, 2020.

\bibitem{Xie.etal2021}
M.~Xie, J.~Liu, Y.~Sun, W.~Gan, B.~Wohlberg, and U.~S. Kamilov,
\newblock ``Joint reconstruction and calibration using regularization by
  denoising with application to computed tomography,''
\newblock in {\em Proc. {IEEE} Int. Conf. Comp. Vis. Workshops ({ICCVW})},
  October 2021, pp. 4028--4037.

\bibitem{Meinhardt.etal2017}
T.~Meinhardt, M.~Moeller, C.~Hazirbas, and D.~Cremers,
\newblock ``Learning proximal operators: {U}sing denoising networks for
  regularizing inverse imaging problems,''
\newblock in {\em Proc. IEEE Int. Conf. Comp. Vis. (ICCV)}, Venice, Italy, Oct.
  22-29, 2017, pp. 1799--1808.

\bibitem{Sun.etal2018a}
Y.~Sun, B.~Wohlberg, and U.~S. Kamilov,
\newblock ``An online plug-and-play algorithm for regularized image
  reconstruction,''
\newblock {\em IEEE Trans. Comput. Imag.}, vol. 5, no. 3, pp. 395--408, Sept.
  2019.

\bibitem{Tirer.Giryes2019}
T.~Tirer and R.~Giryes,
\newblock ``Image restoration by iterative denoising and backward
  projections,''
\newblock {\em IEEE Trans. Image Process.}, vol. 28, no. 3, pp. 1220--1234,
  Mar. 2019.

\bibitem{Teodoro.etal2019}
A.~M. Teodoro, J.~M. Bioucas-Dias, and {M. A. T.} Figueiredo,
\newblock ``A convergent image fusion algorithm using scene-adapted
  {G}aussian-mixture-based denoising,''
\newblock {\em IEEE Trans. Image Process.}, vol. 28, no. 1, pp. 451--463, Jan.
  2019.

\bibitem{Xu.etal2020}
X.~{Xu}, Y.~{Sun}, J.~{Liu}, B.~{Wohlberg}, and U.~S. {Kamilov},
\newblock ``Provable convergence of plug-and-play priors with mmse denoisers,''
\newblock {\em IEEE Signal Process. Lett.}, vol. 27, pp. 1280--1284, 2020.

\bibitem{Sun.etal2021}
Y.~Sun, Z.~Wu, B.~Wohlberg, and U.~S. Kamilov,
\newblock ``Scalable plug-and-play {ADMM} with convergence guarantees,''
\newblock {\em IEEE Trans. Comput. Imag.}, vol. 7, pp. 849--863, July 2021.

\bibitem{Cohen.etal2020}
R.~Cohen, M.~Elad, and P.~Milanfar,
\newblock ``Regularization by denoising via fixed-point projection (red-pro),''
\newblock {\em SIAM Journal on Imaging Sciences}, vol. 14, no. 3, pp.
  1374--1406, 2021.

\bibitem{Tang.Davies2020}
J.~Tang and M.~Davies,
\newblock ``A fast stochastic plug-and-play {ADMM} for imaging inverse
  problems,''
\newblock 2020,
\newblock arXiv:2006.11630.

\bibitem{Monga.etal2021}
V.~Monga, Y.~Li, and Y.~C. Eldar,
\newblock ``Algorithm unrolling: {I}nterpretable, efficient deep learning for
  signal and image processing,''
\newblock {\em IEEE Signal Process. Mag.}, vol. 38, no. 2, pp. 18--44, Mar.
  2021.

\bibitem{Kellman.etal2020}
M.~{Kellman}, K.~{Zhang}, E.~{Markley}, J.~{Tamir}, E.~{Bostan}, M.~{Lustig},
  and L.~{Waller},
\newblock ``Memory-efficient learning for large-scale computational imaging,''
\newblock {\em IEEE Trans. Comp. Imag.}, vol. 6, pp. 1403--1414, 2020.

\bibitem{Liu.etal2021}
J.~Liu, Y.~Sun, W.~Gan, B.~Wohlberg, and U.~S. Kamilov,
\newblock ``Sgd-net: Efficient model-based deep learning with theoretical
  guarantees,''
\newblock {\em IEEE Transactions on Computational Imaging}, vol. 7, pp.
  598--610, 2021.

\bibitem{Wu.etal2019}
Z.~Wu, Y.~Sun, J.~Liu, and U.~S. Kamilov,
\newblock ``Online regularization by denoising with applications to phase
  retrieval,''
\newblock in {\em Proc. {IEEE} Int. Conf. Comput. Vis. Workshops}, Oct. 2019,
  pp. 1--9.

\bibitem{Nesterov2004}
Y.~Nesterov,
\newblock {\em Introductory Lectures on Convex Optimization: A Basic Course},
\newblock Kluwer Academic Publishers, 2004.

\bibitem{Jain.Kar2017}
P.~Jain and P.~Kar,
\newblock ``Non-convex optimization for machine learning,''
\newblock {\em Foundations and Trends in Machine Learning}, vol. 10, no. 3-4,
  pp. 142--363, 2017.

\bibitem{Miyato.etal2018}
T.~Miyato, T.~Kataoka, M.~Koyama, and Y.~Yoshida,
\newblock ``Spectral normalization for generative adversarial networks,''
\newblock in {\em Int. Conf. on Learning Representations ({ICLR})}, Vancouver,
  Canada, Apr. 2018.

\bibitem{anderson1965}
D.~G Anderson,
\newblock ``Iterative procedures for nonlinear integral equations,''
\newblock {\em Journal of the ACM (JACM)}, vol. 12, no. 4, pp. 547--560, 1965.

\bibitem{bai.etal2022}
S.~Bai, V.~Koltun, and J.~Z. Kolter,
\newblock ``Neural deep equilibrium solvers,''
\newblock in {\em International Conference on Learning Representations}, 2022.

\bibitem{aksac2019brecahad}
A.~{Aksac}, D.~J. {Demetrick}, T.~{Ozyer}, and R.~{Alhajj},
\newblock ``Brecahad: a dataset for breast cancer histopathological annotation
  and diagnosis,''
\newblock {\em BMC research notes}, vol. 12, no. 1, pp. 1--3, 2019.

\bibitem{mccollough2016tu}
C.~McCollough,
\newblock ``{TU-FG-207A-04}: Overview of the low dose {CT} grand challenge,''
\newblock {\em Med. Phys}, vol. 43, no. 6Part35, pp. 3759--3760, 2016.

\bibitem{Kingma.Ba2015}
D.~Kingma and J.~Ba,
\newblock ``Adam: {A} method for stochastic optimization,''
\newblock in {\em Proc. Int. Conf. on Learn. Represent.}, 2015.

\bibitem{Lustig.etal2007}
M.~Lustig, D.~L. Donoho, and J.~M. Pauly,
\newblock ``Sparse {MRI}: The application of compressed sensing for rapid {MR}
  imaging,''
\newblock {\em Magn. Reson. Med.}, vol. 58, no. 6, pp. 1182--1195, Dec. 2007.

\bibitem{Lustig.etal2008}
M.~Lustig, D.~L. Donoho, J.~M. Santos, and J.~M. Pauly,
\newblock ``Compressed sensing {MRI},''
\newblock {\em IEEE Signal Process. Mag.}, vol. 25, no. 2, pp. 72--82, Mar.
  2008.

\bibitem{Pruessmann.etal1999}
K.~P. Pruessmann, M.~Weiger, M.~B. Scheidegger, and P.~Boesiger,
\newblock ``{SENSE}: {S}ensitivity encoding for fast {MRI},''
\newblock {\em Magn. Reson. Med.}, vol. 42, no. 5, pp. 952--962, Nov. 1999.

\bibitem{knoll2020fastmri}
{F. {Knoll} \emph{et al.}},
\newblock ``{fastMRI}: A publicly available raw k-space and {DICOM} dataset of
  knee images for accelerated {MR} image reconstruction using machine
  learning,''
\newblock {\em Radiology: Artificial Intelligence}, vol. 2, no. 1, pp. e190007,
  2020.

\bibitem{Ong.etal2019}
F.~Ong and M.~Lustig,
\newblock ``Sigpy: a python package for high performance iterative
  reconstruction,''
\newblock in {\em Proceedings of the ISMRM 27th Annual Meeting, Montreal,
  Quebec, Canada}, 2019, vol. 4819.

\bibitem{Kamilov.etal2016}
U.~S. Kamilov, I.~N. Papadopoulos, M.~H. Shoreh, A.~Goy, C.~Vonesch, M.~Unser,
  and D.~Psaltis,
\newblock ``Optical tomographic image reconstruction based on beam propagation
  and sparse regularization,''
\newblock {\em IEEE Trans. Comp. Imag.}, vol. 2, no. 1, pp. 59--70, Mar. 2016.

\bibitem{Bauschke.Combettes2017}
H.~H. Bauschke and P.~L. Combettes,
\newblock {\em Convex Analysis and Monotone Operator Theory in Hilbert Spaces},
\newblock Springer, 2 edition, 2017.

\bibitem{Ryu.Boyd2016}
E.~K. Ryu and S.~Boyd,
\newblock ``A primer on monotone operator methods,''
\newblock {\em Appl. Comput. Math.}, vol. 15, no. 1, pp. 3--43, 2016.

\bibitem{Rockafellar1970}
R.~T. Rockafellar,
\newblock {\em Convex Analysis},
\newblock Princeton Univ. Press, Princeton, NJ, 1970.

\bibitem{Boyd.Vandenberghe2004}
S.~Boyd and L.~Vandenberghe,
\newblock {\em Convex Optimization},
\newblock Cambridge Univ. Press, 2004.

\bibitem{Wu.etal2018}
Y.~Wu and K.~He,
\newblock ``Group normalization,''
\newblock in {\em Proc. Euro. Conf. Comp. Vis.}, Sep. 2018, pp. 3--19.

\bibitem{Gupta2018}
H.~{Gupta}, K.~H. {Jin}, H.~Q. {Nguyen}, M.~T. {McCann}, and M.~{Unser},
\newblock ``C{NN}-based projected gradient descent for consistent ct image
  reconstruction,''
\newblock {\em IEEE Trans. Med. Imag.}, vol. 37, no. 6, pp. 1440--1453, Jun.
  2018.

\bibitem{Wang.etal2004}
Z.~Wang, A.~C. Bovik, H.~R. Sheikh, and E.~P. Simoncelli,
\newblock ``Image quality assessment: from error visibility to structural
  similarity,''
\newblock {\em IEEE Trans. Image Process.}, vol. 13, no. 4, pp. 600--612, Apr
  2004.

\bibitem{Smith2017}
L.~N. {Smith},
\newblock ``Cyclical learning rates for training neural networks,''
\newblock in {\em 2017 IEEE Winter Conference on Applications of Computer
  Vision}, Mar. 2017, pp. 464--472.


\end{thebibliography}
\end{document}